\documentclass[11pt]{article}
\usepackage{fullpage}

\usepackage[utf8]{inputenc} 
\usepackage[T1]{fontenc}
\usepackage{url}
\usepackage{ifthen}
\usepackage{verbatim}
\usepackage{cite}
\usepackage[cmex10]{amsmath}
\usepackage{amssymb}
\usepackage{amsbsy}
\usepackage{nicefrac}  
\usepackage{braket}
\usepackage{mathpazo}
\usepackage{mathtools}
\usepackage{xcolor}
\usepackage{setspace}
\usepackage[linesnumbered,ruled, lined, commentsnumbered, 
            longend]{algorithm2e}
\usepackage[capitalize,nameinlink]{cleveref}
\crefname{algocf}{alg.}{algs.}
\Crefname{algocf}{Algorithm}{Algorithms}
\renewcommand{\eqref}[1]{\hyperref[#1]{(\ref*{#1})}}
\usepackage{tikz}
\usepackage{pgfplots}
\usepackage{accents}

\newtheorem{theorem}{Theorem}[section]
\newtheorem{corollary}[theorem]{Corollary}
\newtheorem{lemma}[theorem]{Lemma}

\newtheorem{definition}[theorem]{Definition}

\newtheorem{fact}[theorem]{Fact}
\newtheorem{remark}[theorem]{Remark}

\newcounter{example}[section]

\newenvironment{proof}{\textbf{Proof:}}{\hfill$\square$}

\newcommand{\I}{\mathbb{I}}

\newcommand{\brak}[1]{\lbrace #1 \rbrace}

\newcommand{\trho}{\tilde{\rho}}
\newcommand{\tomega}{\tilde{\omega}}

\newcommand{\sub}{\textsubscript}

\newcommand{\gol}[1]{\accentset{\circ}{#1}}

\newcommand{\norm}[1]{\left\lVert#1\right\rVert}

\DeclareMathOperator{\Tr}{Tr}
\DeclareMathOperator{\E}{\mathbb{E}}

\DeclarePairedDelimiter{\abs}{\lvert}{\rvert}

\newcommand{\one}{\leavevmode\hbox{\small1\kern-3.8pt\normalsize1}}

\title{{\bf
One-shot multi-sender decoupling and simultaneous 
decoding for the quantum MAC
}}

\author{
Sayantan Chakraborty${}^*$
\and
Aditya Nema${}^*$
\and
 Pranab Sen\thanks{
School of Technology and System Science, 
Tata Institute of Fundamental Research, Mumbai, India.
Email: {\sf
\{kingsbandz, aditya.nema30, pranab.sen.73\}@gmail.com
}
}
}

\date{}
\usepackage{soul}

\begin{document}
\maketitle

\begin{abstract}
In this work, we prove a novel one-shot ‘multi-sender’ 
decoupling theorem generalising Dupuis’ result. We start off with a 
multipartite quantum state, say on $A_1 A_2 R$, where $A_1$, 
$A_2$ are treated as 
the two ‘sender’ systems and $R$ is the reference system. We apply 
independent Haar random unitaries in tensor product on $A_1$ and $A_2$ and 
then send the resulting systems through a quantum channel. We want the 
channel output $B$ to be almost in tensor with the untouched reference $R$. 
Our main result shows that this is indeed the case if suitable entropic 
conditions are met. An immediate application of our main result is to 
obtain a one-shot simultaneous decoder for sending quantum information 
over a $k$-sender entanglement unassisted quantum multiple access 
channel (QMAC). The rate region achieved by this decoder is the natural 
one-shot quantum analogue of the pentagonal classical rate region. 
Assuming a simultaneous smoothing conjecture, this one-shot rate 
region approaches the optimal rate region of Yard et al. \cite{Yard_MAC} 
in the asymptotic iid limit. 
Our work is the first one to obtain a non-trivial 
simultaneous decoder for the QMAC with limited entanglement assistance
in both one-shot and asymptotic
iid settings; previous works used unlimited entanglement assistance. 
\end{abstract}

\section{Introduction}
The paradigm of {\em decoupling}, that is the process of removing 
correlations between systems, has turned out to be a powerful and
general technique for obtaining inner bounds for transmission of quantum
information in quantum Shannon theory.
Its importance can be seen by its role in obtaining coding strategies
for sending quantum information over a quantum channel, one of the
most basic tasks in quantum Shannon theory. Let 
$\ket{\psi}^{RA}$ be a pure state, where $R$ is the so-called
{\em reference} system that will be untouched by all operations of our
protocol. We want to isometrically encode the {\em message} 
system $A$ into a system $A'G$ and send $A'$ through
a noisy quantum channel $\mathcal{N}^{A'\rightarrow B}$ so that  the
receiver can decode the output $B$ to obtain a state close to 
$\ket{\psi}^{RA}$.
Consider the Stinespring dilation of $\mathcal{N}$, namely 
$\mathcal{U}^{A'\rightarrow BE}_{\mathcal{N}}$, where the system $E$ 
is treated as the purifying {\em environment}. Consider the global 
pure state 
$\ket{\psi}^{RBEG}$. Suppose the following {\em decoupling} condition
holds: $\psi^{REG}\approx 
\psi^R \otimes \sigma^{EG}$ for some state $\sigma$ on $EG$. Let $\I$
denote the identity superoperator. Then by
Uhlmann's theorem one can immediately conclude that 
there exists a decoding {\em isometry} $D^{B\rightarrow A F}$ such that 
\begin{align*}
(D^{B\rightarrow A F} \otimes \I^{REG} \ket{\psi}^{RBEG} \approx 
\ket{\psi}^{RA}\ket{\sigma}^{FEG}, 
\end{align*}
where $\ket{\sigma}^{FEG}$ is a purification of $\sigma^{EG}$.
Thus if a suitable isometric encoder of $A$ into $A'G$ can be found
which satisfies the above decoupling condition, we can achieve quantum
information transmission over a quantum channel without even constructing
an explicit decoder. In other words, decoupling has allowed us to solve
a quantum coding problem doing only half the work as compared 
to the classical setting!

The decoupling paradigm was used to obtain a protocol for the
Fully Quantum Slepian Wolf (FQSW) problem \cite{mother_protocol} 
aka the mother protocol of
quantum Shannon theory as it  can in turn be
used in a black-box fashion to obtain many other protocols for useful
quantum information theoretic tasks in the asymptotic iid setting.
An useful and powerful one-shot decoupling theorem, generalising many
earlier decoupling constructions including that of \cite{mother_protocol},
was obtained by Dupuis \cite{Dupuis_thesis}. We shall refer to this result
henceforth as the {\em single sender decoupling theorem}. 
A high level description follows.
Suppose Alice holds the $A$ 
register of a bipartite mixed state $\rho^{AR}$, where $R$ is the
reference system. Alice applies a Haar random unitary $U$ on $A$ 
followed by a 
completely positive (CP) map $\mathcal{T}^{A\rightarrow E}$. Then
if certain entropic conditions are met, the resulting state typically is
close to the decoupled state $\omega^E \otimes \rho^R$ where $\omega^E$
is a state depending only on the channel $\mathcal{T}$ and not on the
state $\rho^{AR}$ nor on the unitary $U^A$.

The intuition described in the first paragraph of the introduction
above can be made precise and indeed Dupuis' used his single sender 
decoupling theorem to obtain nearly optimal one shot inner bounds for 
sending quantum information over a point to point quantum channel with 
limited entanglement assistance \cite{Dupuis_thesis}, generalising an 
earlier result by
Buscemi and Datta \cite{Buscemi_Datta} for the same problem without
entanglement assistance. In the asymptotic iid setting,
Dupuis' result recovers  the well known regularised
coherent information inner bound when there is no entanglement
assistance \cite{Lloyd, Shor_Notes, Devetak}, and the well known 
mutual information inner bound when there is unlimited entanglment 
assistance \cite{EntanglementAssistedBennett}.

The main contribution of this work is the generalisation of the 
single sender decoupling theorem to the case of independent 
{\em multiple senders}. To 
be precise, we prove a theorem of the following kind. Consider the 
multipartite 
state $\rho^{A_1A_2\ldots A_k R}$ where the users Alice\sub{1}, 
Alice\sub{2} and so on only have access to their respective registers 
$A_1, A_2, \ldots A_k$. Let Alice\sub{i} apply a Haar random
unitary $U_i$ to her register $A_i$ independently of the other Alices.
After the individual unitaries are applied, a CP map 
$\mathcal{T}^{A_1A_2\ldots A_k\rightarrow E}$ is also applied. 
We show, if certain entropic conditions are met, that the resulting
state typically is close to the decoupled state $\omega^E \otimes \rho^R$ 
where $\omega^E$
is a state depending only on the channel $\mathcal{T}$ and not on the
state $\rho^{AR}$ nor on the unitaries $U_i^{A_i}$.

We prove our multi-sender decoupling theorem by suitably extending 
Dupuis' proof of his single sender decoupling theorem.
As will become clear during the course of its proof, one of the main 
bottlenecks in proving such a theorem turns out to be defining and
using the correct one-shot 
entropic quantities in order to bound the error in the protocol.
We show that a modification of the conditional 
R\'{e}nyi $2$-entropy defined
recently in \cite{Nema_Sen} turns out to be the right 
quantity for our purpose. Our simultaneous decoding inner bound is
thus stated in terms of a {\em modified R\'{e}nyi $2$-coherent 
information} and its smoothed version derived from the above quantity.
A similar multi sender decoupling theorem
was earlier proved by Dutil \cite{dutil:phd}. However his formulation
did not use the correct entropic
quantity required to get strong bounds in applications e.g. inner bounds
for quantum multiple access channels.

As an important application of our multi sender decoupling theorem,
we consider the problem of proving inner bounds for the 
Quantum Multiple Access Channel (QMAC) with limited entanglement
assistance in the one-shot setting i.e. when the channel is used only once.
Previous works only considered the QMAC in the asymptotic iid
setting, either with no entanglement assistance \cite{Yard_MAC} or
with unlimited entanglement assistance \cite{hsieh:qmac}. In a
very recent companion paper \cite{chakraborty:ratesplitting}, one shot 
inner bounds were shown for the QMAC with limited entanglement 
assistance which approach the previously known bounds in the 
asymptotic iid setting both without entanglement assistance as well as
with unlimited entanglement assistance. However all these works use
{\em successive cancellation decoding} to obtain their inner bounds. 
Successive cancellation tends to have faster decoding strategies than
{\em simultaneous decoding}. However it also has some drawbacks like
the difficulty of clock synchronisation that arises when used together
with {\em time sharing} in the asymptotic iid setting. This drawback
can be eliminated by a technique called {\em rate splitting} first
developed for the classical asymptotic iid setting by 
\cite{Rate_Splitting_Urbanke} and later extended to the one-shot quantum
setting by \cite{chakraborty:ratesplitting}. However using rate splitting
in the one shot setting brings a new feature which is aesthetically
unappealing viz. the obtained inner bound is a subset of the familiar
`pentagonal' inner bound for the MAC. Note that the pentagonal inner bound
holds both in the classical
asymptotic iid setting \cite{Ahlswede:mac, Liao:mac}, as well as for
transmitting classical information over a quantum MAC both with and
without entanglement assistance \cite{Sen18_2}. Proving a (super) pentagonal
inner bound in the one shot setting requires simultaneous decoding.

Our multi sender decoupling theorem allows us, for the first time,
to get a simultaneous
decoder for sending quantum information over a QMAC with limited
entanglement assistance. This allows us to obtain the 
(super) pentagonal rate region as shown in Figure \ref{fig:pentagon}. 
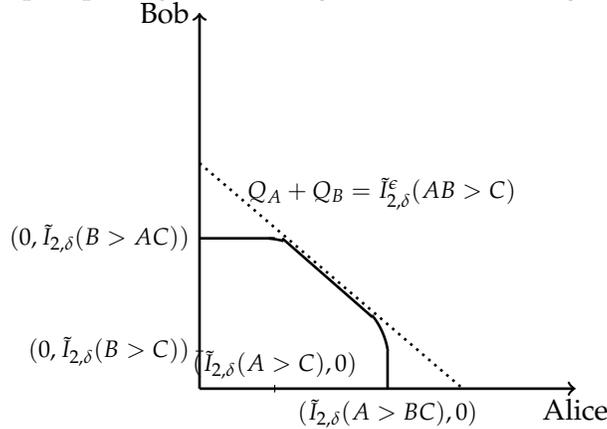
\begin{figure}[hhh]
\vspace*{-5mm}
\begin{center}
\begin{tikzpicture}
\draw [line width=1 pt,->] (1,-2) -- (1,3); 
\draw [line width=1 pt,->] (1,-2) -- (6,-2); 
\draw [line width=1 pt] (3.5,-2) -- (3.5,-1.5); 
\draw [line width=1 pt] (1,0)-- (2,0); 
\draw [line width=1 pt,dotted] (1,1) -- (4.5,-2); 
\draw [shift={(1.85,-1.0)},line width=1 pt]  
plot[domain=1.5:1.3,variable=\t]
({1*cos(\t r)},{1*sin(\t r)});
\draw [shift={(2.515,-1.65)},line width=1 pt]  
plot[domain=0.15:0.68,variable=\t]
({1*cos(\t r)},{1*sin(\t r)});
\draw [line width=1 pt] (2.1,0) -- (3.3,-1.05); 
\draw (1,3) node[anchor=east] {Bob}; 
\draw (6,-2) node[anchor=north] {Alice}; 
\draw (0.95,-1.5) -- (1.05, -1.5);
\draw (1,-1.5) node[anchor=east] 
{{\footnotesize $(0,\tilde{I}_{2,\delta}(B > C))$}}; 
\draw (1,0) node[anchor=east] 
{{\footnotesize $(0,\tilde{I}_{2,\delta}(B > AC))$}}; 
\draw (2,-2.05) -- (2, -1.95);
\draw (2,-2) node[anchor=south] 
{{\footnotesize $(\tilde{I}_{2,\delta}(A > C),0)$}}; 
\draw (3.5,-2) node[anchor=north] 
{{\footnotesize $(\tilde{I}_{2,\delta}(A > BC),0)$}}; 
\draw (1.5,0.6) node[anchor=west] 
{{\footnotesize $Q_A+Q_B= \tilde{I}_{2,\delta}^\epsilon(AB > C)$}}; 
\end{tikzpicture} 
\end{center}
\vspace*{-5mm}
\caption{Inner bound for the unassisted QMAC obtained by simultaneous
decoding. The quantity $\tilde{I}_{2,\delta}(\cdot > \cdot)$ is the 
modified R\'{e}nyi $2$-coherent information and 
$\tilde{I}_{2,\delta}^\epsilon(\cdot > \cdot)$ its smoothed version
defined in Section~\ref{prelims}.
$O(\log \epsilon)$ additive factors have been ignored in the figure.
}
\label{fig:pentagon}
\end{figure}

Simultaneous decoders are an
essential building block for obtaining the best inner bounds for several
multiterminal channels in classical network information theory e.g.
Marton's inner bound with common message for the broadcast channel.
It is expected that our simultaneous decoder for the QMAC will pave the way
for similar results in quantum network information theory too.

A shortcoming of our results is that we are unable to show that our
one shot inner bound for the unassisted QMAC recovers the optimal
asymptotic iid result of \cite{Yard_MAC}. To do so, one would require 
the existence of a single state which simultaneously smooths and 
(nearly) maximises all the 
three entropic quantities that arise in our simultaneous decoding
inner bound picturised in Figure~\ref{fig:pentagon}. The existence of
such a state is a major open problem in quantum information theory and 
is known 
as the {\em simultaneous smoothing conjecture}. The interested reader 
is referred to \cite{Sen18_1} for more details.

In this paper we will be using a variant of the $2$ entropy, defined by 
Nema and Sen in \cite{Nema_Sen}. The interested reader is referred to 
\cite{Tomamichel_thesis} for an extremely comprehensive survey of these 
quantities and their properties.
A n\"aive choice of the smoothed conditional R\'enyi $2$-entropy does 
not work, for technical reasons that we will describe in 
\cref{multidecoupthm}.

\subsection{Organisation of the Paper}
The rest of the paper is organised as follows:
\begin{itemize}
\item In \cref{prelims} we define the one-shot entropic quantities that 
we require to prove our theorems, along with the statements of useful 
facts about these quantities and other identities in general. 
\item In \cref{multidecoupthm} we state and prove the multi sender 
decoupling theorem. 
\item In \cref{QMAC} we use the $2$ sender version of the multi 
sender decoupling theorem to derive inner bounds for sending quantum 
information via the QMAC.
\item We conclude with \cref{Conclusion and open problem} by mentioning 
an immediate open problem which might be useful for other problems in 
quantum Shannon theory.
\end{itemize}
 
\section{Preliminaries}
\label{prelims}
\subsection{Notation}
All vector spaces considered the paper are finite dimensional inner product
spaces, also called finite dimensional Hilbert spaces, over the complex 
field and denoted by $\mathcal{H}$. 
We use $|\mathcal{H}|$ to denote the dimension of a Hilbert space 
$\mathcal{H}$. 
Logarithms are all taken in base two.
We tacitly assume that the ceiling is taken of any formula that provides 
dimension or value of $t$ in unitary $t$-design.
The symbols $\mathbb{E}$, $\mathbb{P}$ denote expectation and probability 
respectively. The abbreviation "iid" is used to mean identically and 
independently distributed, which just means taking the tensor power of 
the identical copies of the underlying state. 
The notation ":=" is used to denote the definitions 
of the underlying mathematical quantities. 

The notation
$\mathcal{L}(A_1,A_2)$ denotes the Hilbert space of 
all linear operators from Hilbert space $A_1$ to Hilbert space
$A_2$ with the inner product being the Hilbert-Schmidt inner product
$\langle X, Y \rangle := \Tr [X^\dag Y]$.
For the special case when $A_1 = A_2$ we use
the phrase operator on $A_1$ and the symbol 
$\mathcal{L}(A_1)$.  
The symbol $\mathbb{I}^{A}$ denotes the 
identity operator on vector space $A$.
The matrix $\pi^A$ denotes the so-called completely mixed 
state on system $A$, i.e., $\pi^A := \frac{\mathbb{I}^A}{|A|}$.
We use the notation $U \circ A$ 
as a short hand to denote the conjugation of the operator $U$ on the 
operator $M$, that is, $U \cdot M := U M U^\dagger$.

The symbol $\rho$ usually denotes a quantum state, also called as 
a density matrix 
which is a Hermitian positive semidefinite matrix with
unit trace, and
$\mathcal{D}(\mathbb{C}^d)$ denotes the set of all $d \times d$ density 
matrices. The symbol $\mathrm{Pos}(\mathbb{C}^d)$ denotes the set of all
$d \times d$ positive semidefinite matrices, and the symbol
$\mathbb{U}(d)$ denotes the set of all $d \times d$ unitary matrices
with complex entries.
For a positive semidefinite matrix $\sigma$, we use $\sigma^{-1}$ to
denote the operator which is the orthogonal direct sum of
the inverse of $\sigma$ on its support and the zero operator on the
orthogonal complement of the support. This definition of 
$\sigma^{-1}$ is also known as the {\em Moore-Penrose pseudoinverse}.
The symbol $\ket{\psi}$ denotes a vector $\psi$ of unit Schatten $2$-norm 
or the Frobenius norm, 
and $\bra{\psi}$ denotes the corresponding linear functional.
A pure quantum state is rank one density matrix.
For brevity, a pure quantum state 
$\ket{\psi}\bra{\psi}$ is denoted by $\psi$ to emphasise that it is 
a density matrix. For 
two Hermitian matrices $A$, $B$ of the same dimension, we use 
$A \geq B$ as a shorthand to imply that the matrix $A - B$ is positive
semidefinite.

Let $X \in \mathcal{L}(A)$.
The symbol $\Tr X$ denotes the 
trace of operator $X$. Trace is a linear map from $\mathcal{L}(A)$
to $\mathbb{C}$.
Let $A$, $B$ be two vector spaces.
The partial trace $\Tr_A [\cdot]$ obtained by tracing out $A$ 
is defined to be the unique linear map from 
$\mathcal{L}(A \otimes R)$ to $\mathcal{L}(R)$ satisfying 
$\Tr_A [X \otimes Y] = (\Tr X) Y$ for all operators 
$X \in \mathcal{L}(A)$, $Y \in \mathcal{L}(R)$.

A linear map $\mathcal{T}: \mathcal{M}_m \to \mathcal{M}_d $, that 
maps a linear operator to another linear operator is called a 
superoperator. A superoperator $\mathcal{T}$ is said to be 
{\em positive} if 
it maps 
positive semidefinite operators to positive semidefinite operators,
and {\em completely positive} if $\mathcal{T} \otimes \mathbb{I}$ is a 
positive superoperator for all identity superoperators $\mathbb{I}$.
A superoperator $\mathcal{T}$ is said to be {\em trace preserving} if 
$\Tr [\mathcal{T}(X)] = \Tr [X]$ for all $X \in \mathcal{M}_m$.
Completely positive and Trace Preserving (abbreviated as CPTP) 
superoperators are called 
{\em quantum operations} or {\em quantum channels}.
In this paper all the superoperators considered are completely positive 
and trace 
non-increasing superoperators, unless stated otherwise. The symbol 
$\mathbb{I}$ denotes the identity or the noiseless channel which does 
not alter the input at all or $\mathbb{I}(X)=X$, for all operators $X$. 

The adjoint of a superoperator is defined with respect to the 
Hilbert-Schmidt inner product on matrices. If
$\mathcal{T}: \mathcal{M}_m \to \mathcal{M}_d $ is a superoperator, then 
its adjoint
$\mathcal{T}^\dag: \mathcal{M}_d \to \mathcal{M}_m $ is a 
superoperator uniquely
defined by the property that
$
\langle \mathcal{T}^\dag(A), B \rangle =
\langle A, \mathcal{T}(B) \rangle 
$
for all $A \in \mathcal{M}_d$, $B \in \mathcal{M}_m$.

\subsection{Entropic Quantities}
In this section we define the relevant entropic quantities used in the 
proof of our general multi-user decoupling theorem. We start with the 
conditional min entropy, followed by conditional $2$-entropy, a variant 
of which will be used in most of our proofs. 
\begin{definition}
\label{def:Hmin}
Let $0 \leq \epsilon < 1$. The  $\epsilon$-smooth conditional 
min-entropy of 
$\rho^{AB}$ is defined as:
\[
H_{\mathrm{min}}^\epsilon(A|B)_\rho :=
\min_{
\substack{ 
\sigma^{AB} \in \mathrm{Pos}(B): \\
\lVert \sigma^{AB} - \rho^{AB} \rVert_1 \leq \epsilon
}} \,
\{ \Tr(\gamma^B): \gamma^B \in Pos(B), 
\sigma^{AB} \leq I^A \otimes \gamma^B \}.
\]
When $\epsilon=0$, this is just $H_{min}(A|B)$ with $\sigma^{AB}$ 
replaced by $\rho^{AB}$.
\end{definition}

\begin{definition}
\label{def:Renyi2}
Let $0 \leq \epsilon < 1$. 
The $\epsilon$-smooth conditional R\'{e}nyi 2-entropy for a 
bipartite positive 
semidefinite operator $\rho^{AR}$ on systems $A$ and $R$ is defined as:
$$
H_2^\epsilon(A|R)_{\rho} := 
-2 \log 
\min_{
\substack{ 
{\sigma^{AR} \in \mathrm{Pos}(AR):}\\
{\lVert \rho^{AR}-\sigma^{AR}\rVert_1 \leq \epsilon} \\
{\omega^R \in \mathcal{D}(R): \omega^R > 0^R}
}
}
\{ 
\lVert 
(\omega^R \otimes I^A)^{-1/4} \sigma^{AR} (\omega^R \otimes I^A)^{-1/4} 
\rVert_2
\}.
$$ 
When $\epsilon = 0$, 
we simply refer to the above quantity as conditional 
R\'{e}nyi $2$-entropy and denote it by $H_2(A|R)_{\rho}$ and
define 
$
\trho^{AR} :=
(\omega^R \otimes I^A)^{-1/4} \rho^{AR} (\omega^R \otimes I^A)^{-1/4}.
$
\end{definition}
The advantage of working with smoothed conditional R\'enyi $2$-entropy 
is that in the asymptotic iid limit it is appropriately bounded by 
conditional Shannon entropy, as mentioned in the following 
Fact~\ref{fact:Renyi2upperbound}:
\begin{fact}
\label{fact:Renyi2upperbound}
Let $\epsilon > 0$. Then,
$
H_2^\epsilon(A|B)_\omega  \leq 
H(A|B)_\omega + 8 \epsilon \log |A| + 
2 + 2 \log \epsilon^{-1}.
$ and
$
H_{min}^\epsilon(A|B)_\rho \geq 
H(A|B)-8 \log |A| \times \sqrt{\log \frac{2}{\epsilon^2}}
$
\end{fact}
The proof of the bounds on $H_{min}$ can be found in 
\cite[Theorem~9]{QAEP} and for $H_2$ can be deduced by combining 
\cite[Equation~8]{tomamichel:Renyi} with 
\cite[Theorem~7, Lemma~2, Equation~33]{QAEP} and then
applying the Alicki-Fannes inequality \cite{AlickiFannes}, respectively .

For our proofs we will be using a slightly modified version of the 
$2$-entropy, where we fix the $\sigma^B$ to a special state instead 
of optimising over it. The justification for this definition is as follows:
\begin{enumerate}
\item This quantity is much more tractable than the optimised $2$-entropy.
\item The smoothed version of this new quantity indeed approaches the 
conditional Shannon entropy in the asymptotic iid limit, as proved in 
\cite{Nema_Sen}.  
\end{enumerate}
\begin{definition}{\textbf{$\delta$-Tilde Conditional $2$-Entropy}}
Given a state $\rho^{AB}$ on the registers $AB$ and $\delta\in (0,1)$ the $\delta$-Tilde conditional $2$-entropy of $A$ given $B$ is defined as 
\begin{align*}
\Tilde{H}_{2,\delta}(A|B)_\rho\coloneqq 
-\log \Tr[\big((\mathbb{I}^A\otimes \rho^{B}_{\delta})^{-1/2} 
	  \rho^{AB})\big)^2]
\end{align*}
where $\rho^{B}_{\delta}$ is that positive semidefinite matrix that is 
obtained by zeroing out the smallest eigenvalues of $\rho^B$ that sum 
to less than or equal to $\delta$.
\end{definition}
The smoothed variant of this quantity, as defined by Nema and Sen was 
shown in \cite{Nema_Sen} to approach the Shannon conditional entropy 
in the asymptotic iid limit. We call this the $\epsilon$-smooth 
$\delta$-tilde conditional $2$-entropy. In the interest of brevity, we 
will refer to this quantity simply as the smooth tilde $2$-entropy from 
now on. We will require some additional definitions before introducing 
this quantity:
\begin{definition}{\textbf{$\epsilon$-smooth Max Entropy}}
Given a state $\rho^A$ and positive $\epsilon$, the 
$\epsilon$-smooth Max Entropy max entropy is given by
\begin{align*}
H^{\epsilon}_{\max}(A)_{\rho}\coloneqq 
2\log\min\limits_{\substack{\rho'\geq 0 \\ \norm{\rho'-\rho}_1\leq 
\epsilon}}\Tr[\sqrt{\rho'}]
\end{align*}
\end{definition}
\begin{definition}{\textbf{$\delta$-Tilde Max Entropy }}
Given the state $\rho^A$, consider the state $\rho_{\delta}$ which is 
obtained by zeroing out the smallest eigenvalues of $\rho^A$ which sum 
to less than or equal to $\delta$. Then the $\delta$-Tilde Max Entropy 
is given by 
\begin{align*}
\Tilde{H}^{\delta}_{\max}(A)_{\rho}\coloneqq 
\log\norm{\big(\rho^B_{\delta}\big)^{-1}}_{\infty}
\end{align*}
\end{definition}
We are now ready to define the smooth tilde $2$-entropy:
\begin{definition}{\textbf{$\epsilon$-smooth $\delta$-Tilde 
Conditional $2$-Entropy}}
Given a state $\rho^{AB}$, consider the state 
$\rho'^{B}_{\epsilon, \delta}$ that is obtained by zeroing out those 
eigenvalues of $\rho^B$ which are smaller than 
$2^{-(1+\delta)\Tilde{H}^{\epsilon}_{\max}(B)_{\rho}}$ . Then, we define 
\begin{align*}
\Tilde{H}^{\epsilon}_{2,\delta}(A|B)_{\rho}\coloneqq 
-\log \min\limits_{\substack{0\leq \eta^{AB}\leq \rho^{AB} \\ 
\norm{\eta-\rho}_1\leq \epsilon}} 
\Tr[\big((\mathbb{I}^A\otimes \rho'^{B}_{\epsilon,\delta})^{-1/2} 
	\rho^{AB})\big)^2]
\end{align*}

\end{definition}
\begin{fact}
For $n\in \mathbb{N}$ and $\epsilon,\delta>0$, given a quantum state 
$\rho^{AB}$ and its iid extension ${\rho^{AB}}^{\otimes n}$ the following 
holds
\begin{align*}
\lim\limits_{\substack{\epsilon,\delta\rightarrow 0}}
\lim\limits_{n\rightarrow \infty} 
\frac{\Tilde{H}^{\epsilon}_{2,\delta}(A^n|B^n)_{\rho^n}}{n} 
\geq H(A|B)_{\rho}
\end{align*}
\end{fact}
\subsection{Useful Facts}
\begin{fact}\mbox{\cite{Watrous_notes, Nema_Sen}}
\label{fact:sStinespring}
Any superoperator $\mathcal{T}^{A \to B}$ can be represented as:
$$
\mathcal{T}^{A \to B}(M^A) = 
\Tr_Z \{ 
V_{\mathcal{T}}^{AC \to BZ} 
(M^A \otimes (\ket{0}\bra{0})^C) 
(W_{\mathcal{T}}^{AC \to BZ})^{\dagger} 
\} 
$$
where $V_{\mathcal{T}}$, $W_{\mathcal{T}}$ are operators that map 
vectors from  
$A \otimes C$ to vectors in $B \otimes Z$. Systems 
$C$ and $Z$ are considered as the input and output ancillary 
systems respectively, such that $|A| |C| = |B| |Z|$. Without loss of
generality, $|C| \leq |B|$ and $|Z| \leq |A|$.
Furthermore, in the following special cases $V_{\mathcal{T}}$, 
$W_{\mathcal{T}}$ have additional properties.
\begin{enumerate}

\item 
$\mathcal{T}$ is completely positive if and only if 
$V_{\mathcal{T}} = W_{\mathcal{T}}$.

\item 
$\mathcal{T}$ is trace preserving if and only if 
$V_{\mathcal{T}}^{-1} = W_{\mathcal{T}}^\dag$.
Thus, $\mathcal{T}$ is completely positive and trace preserving if and 
only if 
$V_{\mathcal{T}} = W_{\mathcal{T}}$ and are unitary operators.

\item 
$\mathcal{T}$ is completely positive and trace non-decreasing if and 
only if 
$V_{\mathcal{T}} = W_{\mathcal{T}}$ and 
$\lVert V_{\mathcal{T}} \rVert_\infty \leq 1$. 

\end{enumerate}  
\end{fact}

\begin{fact}\cite{Dupuis_thesis}
\label{swaptrick}
{\textbf{Swap Trick}} Given two linear operators $M$ and $N$ on the 
system $A$ and the swap operator $F^{AA'}$, where we denote by $A'$ 
an isomorphic copy of $A$, the following holds:
\begin{align*}
    \Tr[MN]=\Tr[\big(M\otimes N\big)F^{AA'}]
\end{align*}
\end{fact}

\begin{fact}\mbox{\cite{Watrous_notes}}
\label{fact:Uhlmann}
{\textbf{Uhlmann's Theorem}} For quantum states $M$ and $N$ with 
purifications $\ket{\phi}^{XY}$ and $\ket{\psi}^{XZ}$ respectively 
(referring to the systems $Y, \; Z$ as purifying systems and systems 
$Y,Z$ need not be isomorphic). Then,
\[
F(M,N)=\max_{V^{Y \to Z}} |\bra{\phi}V^\dag \ket{\psi}|
\]
where the maximization is over all partial isometries 
$V (\equiv V^\dag V=I^Y)$ from $Y$ to $Z$ with $dim(Z) \geq dim(Y)$.
\end{fact}

\begin{fact}\label{tensor integ}
Given a linear operator $M$ on $A^{\otimes 2}$ it holds that 
\begin{align*}
\mathbb{E}(M)=\int\limits_{U\in \mathbb{U}(A)}
\big( U^{\otimes 2}\cdot M\big)dU = \alpha \mathbb{I}^{AA'}+\beta F^{AA'}
\end{align*}
where $\alpha$ and $\beta$ are the solutions of the equations 
$\Tr[M]=\alpha\abs{A}^2+\beta\abs{A}$ and 
$\Tr[FM]=\alpha \abs{A}+\beta\abs{A}^2$ and the integration is over the 
Haar measure over the Unitary group.
\end{fact}
\begin{fact}\label{boundlem}
Given a positive semidefinite operator $\rho^{AB}$ on the system $AB$ 
\begin{align*}
\frac{1}{\abs{A}}\leq 
\frac{\Tr[{\rho^{AB}}^2]}{\Tr[{\rho^{B}}^2]}\leq \abs{A}
\end{align*}
\end{fact}
\begin{fact}\label{cauchyschwartz}
Let $M$ be any linear operator and $\sigma$ be a positive semidefinite 
operator on system $A$. Then
\begin{align*}
\norm{M}_1\leq \sqrt{\Tr[\sigma]
\Tr[\sigma^{-1/4}M\sigma^{-1/2}M^{\dagger}\sigma^{-1/4}]}
\end{align*}
and in particular, when $M$ is Hermitian
\begin{align*}
\norm{M}_1\leq \sqrt{\Tr[\sigma]
\Tr[\big(\sigma^{-1/4}M\sigma^{-1/4}\big)^2]}
\end{align*}
\end{fact}

\section{Inner Bounds for the 
QMAC using the Multi Sender Decoupling Theorem}

We will first consider the task of entanglement transmission. As before, 
we first consider the seemingly more general problem: We are given a 
QMAC $\mathcal{N}^{A'B'\rightarrow C}$ and two states pure states 
$\psi^{AC_1R_1}$ and $\varphi^{BC_2R_2}$, with Alice holding the system 
$A$, Bob the system $B$ and Charlie the systems $C_1C_2$. $R_1R_2$ are 
the reference registers. Alice and Bob wish to send the registers $A$ 
and $B$ to Charlie through one use of the channel $\mathcal{N}$, such 
that at the end of the protocol, the state that Charlie holds is close 
to $\psi^{AC_1R_1}\varphi^{BC_2R_2}$. To do this, we must show the 
existence of encoders $\mathcal{E}_1$ and $\mathcal{E}_2$ and a 
decoder $\mathcal{D}$ such that

\begin{align*}
\norm{\mathcal{D}\circ \mathcal{N}\circ \big(\mathcal{E}\otimes 
\mathcal{E}_2\big)\big(\psi\otimes \varphi\big)-
\psi\otimes \varphi}_1\leq \epsilon
\end{align*}

We consider the complementary channel 
$\bar{\mathcal{N}}^{A'B'\rightarrow E}$ and the randomized encoders 
$\mathcal{E}_1$ and $\mathcal{E}_2$ such that

\begin{align}\label{needcultidecoup}
\bar{N}^{A'B'\rightarrow E}\circ 
\big(\mathcal{E}^{A\rightarrow A'}_1\otimes 
\mathcal{E}^{B\rightarrow B'}_2\big) 
\big(\psi^{AR_1}\otimes \varphi^{BR_2}\big) 
\approx 
\bar{N}^{A'B'\rightarrow E}\circ 
\big(\mathcal{E}^{A\rightarrow A'}_1\otimes 
\mathcal{E}^{B\rightarrow B'}_2\big) \big(\psi^{A}\otimes 
\varphi^{B}\big)\otimes \big(\psi^{R_1}\otimes \varphi^{R_2}\big)
\end{align}

We will first fix a pure control state 
$\ket{\omega}^{A"B"CE}\coloneqq 
\mathcal{U}^{A'B'\rightarrow CE}_{\mathcal{N}}
\ket{\Omega}^{A"A'}\ket{\Delta}^{B"B'}$. We consider randomized encoders 
$\mathcal{E}_1$ and $\mathcal{E}_2$, where the randomness is derived 
from independently picked unitaries $U_1$ and $U_2$, each of which 
is identically distributed with respect to the Haar measure. The single 
user decoupling theorem will clearly not work here, and hence we use 
our multisender decoupling theorem instead. Using that theorem, we show 
that there exist decoders which obey \cref{needcultidecoup}, as long 
as  the following entropic inequalities are satisfied:

\begin{enumerate}
\item 
$-\tilde{H}_{2,\delta}(A"B"|E)-
\tilde{H}_{2,\delta}(A|R_1)_{\psi}-
\tilde{H}_{2,\delta}(B|R_2)_{\varphi}\leq \log\epsilon$
\item 
$-\tilde{H}_{2,\delta}(B"|E)_{\omega}-
\tilde{H}_{2,\delta}(B|R_2)_{\varphi}\leq \log\epsilon$
\item 
$-\tilde{H}_{2,\delta}(A"|E)_{\omega}-
\tilde{H}_{2,\delta}(A|R_1)_{\varphi}\leq \log \epsilon$
\end{enumerate}
 
One should note that to finish the argument, one still has to show the 
existence of two \emph{fixed} isometric encoders $V_1$ and $V_2$ which 
perform almost as well as the randomized decoders. This argument however 
does not require the power of the multi sender decoupling theorem, a 
two separate applications of the single sender decoupling theorem provide 
a proof, with error estimates in terms of 
$H_{\max}(A)_{\psi}-\Tilde{H}_{2,\delta}(A")_{\omega}$ and 
$H_{\max}(B)_{\varphi}-\Tilde{H}_{2,\delta}(B")_{\omega}$. The reader is 
referred to \cref{dupuismac} for a detailed proof of the claims made above.

It is now easy to show inner bounds for the entanglement transmission for 
the QMAC from these bounds. We simply set $\psi^{AC_1R_1}$ to be 
$\Phi^{R_1M_1}\otimes \Phi^{\Tilde{A}C_1}$ and 
$\varphi^{BC_2R_2}$ to be $\Phi^{R_2M_2}\otimes \Phi^{\tilde{B}C_2}$. 
Recall that $\Phi^{\Tilde{A}C_1}$ and $\Phi^{\tilde{B}C_2}$ represent 
pre-shared entanglement. So the inner bounds we derive are for 
\emph{partial} entanglement assistance. For the unassisted inner bounds 
we simply set the systems $\Tilde{A}$ and $\Tilde{B}$ to be trivial. 
Please refer to \cref{thm:QMAC} for details.

We note that it is possible to design an entanglement generation protocol 
for the QMAC as well using the multi sender decoupling theorem. The idea 
is as follows: we consider the control state 
$\ket{\omega}^{ABCE}\coloneqq 
\mathcal{U}^{A'B'\rightarrow CE}_{\mathcal{N}}
\ket{\Omega}^{AA'}\ket{\Delta}^{BB'}$. Consider the projectors 
$\Pi^{A\rightarrow R_1}_1$ and $\Pi^{B\rightarrow R_2}_2$, of rank 
$\abs{R_1}$ and $\abs{R_2}$, where $R_1$ and $R_2$ are subspaces of $A$ 
and $B$ respectively. The idea is we hit the systems $A$ and $B$ with 
the operators $\frac{\abs{A}}{\abs{R_1}}\Pi^{A\rightarrow R_1}_1 U^{A}$ 
and $\frac{\abs{B}}{\abs{R_2}}\Pi^{A\rightarrow R_2}_2 U^{B}$ 
respectively, where $U^A$ and $U^B$ are Haar random unitaries. We want to 
show that, on average, the systems $R_1$ and $R_2$ are decoupled from 
the system $E$. 

We note that, in contrast to the case of entanglement transmission, 
where the random unitary is a part of the encoder and acts on the system 
to be transmitted, in the case of entanglement generation, the averaging 
is actually done of the action of the random unitaries on the 
\emph{purifying system}. Because of this, the case for entanglement 
generation seems potentially more challenging, as, instead on acting 
on two tensored subsystems ($\psi^{AR_1}\otimes \varphi^{BR_2}$), the 
random unitaries act on a state which are entangled via the system 
$E$ viz. $\omega^{ABE}$. 

Nonetheless, our multisender decoupling theorem is general enough to 
handle this case as well (see \cref{twodecoup}) and we are indeed able 
to show that on average the systems $R_1$ and $R_2$ are decoupled with $E$.

The proofs of the claims made above are included in the Appendix.

\section{The Multi-Sender Decoupling Theorem}\label{multidecoupthm}
Before we move on to the general multi sender decoupling theorem it will 
be instructive to see the proof for the case of only $2$ senders. 
The proof for more than $2$ senders requires heavy notation. 
Hence we defer its exposition to a later section.

\subsection{A Warm-Up: The 2-Sender Decoupling Theorem}
\textbf{Notation :}  We will sometimes abbreviate the symbol for the 
unitary $U^{A_i}$ as $U_i$ to ease the notation. As before we will 
use the $'$ accent in conjunction with the name of a register to 
denote isomorphic copies of the original system, for e.g. $A$ and $A'$.\\
\vspace{2mm}\\
\noindent We will require the following lemma:
\begin{lemma}\label{schurweylgeneral}
Given a linear operator $M$ on the space $A_1A_1'A_2A_2'$ the following 
holds
\begin{align*} 
&=\int\big(U_1^{\otimes 2}\otimes U_2^{\otimes 2}\big)\cdot M~dU_1dU_2 \\
&= \alpha_{00}\mathbb{I}^{A_1A_1'}\otimes \mathbb{I}^{A_2A_2'}
	+\alpha_{01}F^{A_1A_1'}\otimes \mathbb{I}^{A_2A_2'}\\ 
&+\alpha_{10}\mathbb{I}^{A_1A_1'}\otimes F^{A_2A_2'}+
	\alpha_{11}F^{A_1A_1'}\otimes F^{A_2A_2'}
\end{align*}
where the coefficients $\alpha_{ij}$ are given  by 
\begin{align*}
    \abs{A_1A_2}~
    \begin{bmatrix}
    \abs{A_2} & 1 \\
    1 & \abs{A_2}
    \end{bmatrix} \bigotimes
    \begin{bmatrix}
    \abs{A_1} & 1 \\
    1 & \abs{A_1}
    \end{bmatrix}
    \begin{bmatrix}
    \alpha_{00} \\ \alpha_{01} \\ \alpha_{10} \\ \alpha_{11}
    \end{bmatrix}
    =  
    \begin{bmatrix}
\Tr[M] \\ 
\Tr[F^{A_1A_1'}\otimes \mathbb{I}^{A_2A_2'}M] \\ 
\Tr[\mathbb{I}^{A_1A_1'}\otimes F^{A_2A_2'}M] \\ 
	   \Tr[F^{A_1A_1'}\otimes F^{A_2A_2'}M]
    \end{bmatrix}
\end{align*}
\end{lemma}
\begin{proof}
The proof is an easy extension of \cref{tensor integ} and the linearity 
of integration. We give it for completeness.\\
We expand $M^{A_1A_1'A_2A_2'}$ in a Schmidt decomposition as :
$M=\Sigma_i \left(X_i^{A_1A_1'} \otimes (X'_i)^{A_2A_2'} \right)$. Then:
\begin{align*}
\int \left( 
\big[ U_1^{\otimes 2} \otimes U_2^{\otimes 2} \big] \cdot M  \right) 
dU_1 dU_2 
&= \int \left( \Sigma_i \big[ U_1^{\otimes 2} \otimes U_2^{\otimes 2} \big]
\cdot \{ X_i^{A_1A_1'} \otimes (X'_i)^{A_2A_2'} \}  \right) dU_1 dU_2  \\
&= \Sigma_i \left[ \int \left( \big[ U_1^{\otimes 2} \cdot 
   X_i^{A_1A_1'}\big] \otimes \big[U_2^{\otimes 2}  \cdot 
	\otimes (X'_i)^{A_2A_2'} \big]  \right) dU_1 dU_2 \right]\\
&= \Sigma_i \left[ \int \left(  U_1^{\otimes 2} \cdot X_i^{A_1A_1'} \right)
	dU_1 \otimes 
\int \left( U_2^{\otimes 2}  \cdot \otimes (X'_i)^{A_2A_2'}  \right)dU_2 
	\right] \\
& \overset{a}{=} \Sigma_i \left[ \left( \alpha_i I^{A_1 A_1'} + 
\beta_i F^{A_1 A_1'} \right) \otimes \left( \alpha'_i I^{A_2 A_2'} + 
\beta'_i F^{A_2 A_2'} \right) \right]\\
& \overset{b}{=} \alpha_{00} I^{A_1A_1'} \otimes I^{A_2A_2'} + 
\alpha_{01} F^{A_1A_1'} \otimes I^{A_2A_2'} + 
\alpha_{10} I^{A_1A_1'} \otimes F^{A_2A_2'} + 
\alpha_{11} F^{A_1A_1'} \otimes F^{A_2A_2'}
\end{align*}
where,
\begin{itemize}
\item{(a)} follows by \cref{tensor integ}.
\item{(b)} holds by the identification that 
$\alpha_{00}:= \Sigma_i \alpha_i, \; \alpha_{01}:= \Sigma_i \beta_i, \; 
 \alpha_{10}:= \Sigma_i \alpha'_i \text{ and } 
\alpha_{11}:= \Sigma_i \beta'_i$. 
\end{itemize}
In order to obtain the values of the coefficients 
$\{ \alpha_{kl} \}_{k,l = 0}^{1}$, we use the identities 
$\Tr \left[ \int U_1^{\otimes 2} \cdot X_i dU_1  \right] = 
\Tr[X_i] = \alpha_i |A_1|^2 + \beta_i |A_1|$ and 
$\Tr \left[F \int U_1^{\otimes 2} \cdot X_i dU_1 \right]=
 \Tr[FX_i]=\alpha_i |A_1|+\beta_i |A_1|^2$ 
in equality (b) and solving the system of equations simultaneously, 
completes the proof. 
\end{proof}

We now state and prove the main technical tool of this work, the 
$2$ sender decoupling theorem:

\begin{theorem}\label{twodecoup}
Let $\rho^{A_1A_2R}$ be a density operator, 
$\mathcal{T}^{A_1A_2\rightarrow E}$ be a CP map, and define 
$\omega^{A'_1A'_2E}\coloneqq 
(\mathbb{I}^{A_1A_2}\otimes\mathcal{T})(\Phi^{A_1A'_1}\otimes 
\Phi^{A_2A'_2})$. For a given $\delta>0$, we define 
$\sigma^E\coloneqq \omega^E_{\delta}$ as the positive semidefinite matrix 
obtained by zeroing out the smallest eigenvalues of $\omega^E$ that 
sum to at most $\delta$. 
We define $\zeta\coloneqq \rho^{R}_{\delta}$ similarly. 

Define 
\begin{enumerate}
\item $\Tilde{\mathcal{T}}\coloneqq {\sigma^E}^{-1/4}\cdot \mathcal{T}$
\item $\Tilde{\rho}^{A_1A_2R}\coloneqq 
	{\zeta^R}^{-1/4}\cdot \rho^{A_1A_2AR}$
\item $\Tilde{\omega}^{A'_1A'_2E}\coloneqq 
	\Tilde{\mathcal{T}}(\Phi^{A_1A'_1}\otimes \Phi^{A_2A'_2})$  
\end{enumerate}

Then,
\begin{align*}
&\int\norm{\mathcal{T}\{ (U^{A_1}\otimes U^{A_2})\cdot \rho^{A_1A_2R}) \}
-\omega^E\otimes \rho^{R}}_1 dU^{A_1}dU^{A_2} \\
& \leq\Big(\delta+\frac{|A_2|^2}{|A_2|^2-1} 
\Big\{\norm{\trho^{A_1R}}_2^2\norm{\tomega^{A'_1E}}_2^2+
	\norm{\trho^{A_2R}}_2^2\norm{\tomega^{A'_2E}}_2^2+
	\norm{\trho^{A_1A_2R}}_2^2
	\norm{\tomega^{A'_1A'_2E}}_2^2\Big\}\Big)^{1/2}
\end{align*}
where $|A_2|>|A_1|$.
\end{theorem}

\begin{proof}
We will start by using \cref{cauchyschwartz} with the weighting matrix 
$(\sigma^E\otimes \zeta^R)$. First observe that by definition 
$\Tr[\sigma^E\otimes \zeta^R]\leq 1$. Then,
\begin{align*}
\int\norm{\mathcal{T} \{ (U^{A_1}\otimes U^{A_2})\cdot \rho^{A_1A_2R} \}
-\omega^E\otimes \rho^{R}}_1 
&\leq 
\sqrt{\int \Tr\bigg[\big(\Tilde{\mathcal{T}} 
\{(U^{A_1}\otimes U^{A_2})\cdot \trho^{A_1A_2R}\}-
  \tomega^E\otimes \trho^{R}\big)^2\bigg]}
\end{align*}
Next, opening up the square and using the fact that the integral and 
$\Tilde{\mathcal{T}}$ commute we see that,
\begin{align*}
&\int \Tr\bigg[\big(\Tilde{\mathcal{T}}\{(U^{A_1}\otimes U^{A_2})\cdot 
\trho^{A_1A_2R}\}-\tomega^E\otimes \trho^{R}\big)^2\bigg] \\
&= \int \Tr\bigg[\big(\Tilde{\mathcal{T}}\{(U^{A_1}\otimes U^{A_2})\cdot 
\trho^{A_1A_2R}\}\big)^2\bigg]-
2 \Tr\Bigg[\Tilde{\mathcal{T}}\Big\{\int (U^{A_1}\otimes U^{A_2})\cdot 
\trho^{A_1A_2R}dU_1dU_2\Big\}(\tomega^E\otimes \trho^{R})\bigg] \\
& +\Tr\bigg[(\tomega^E\otimes \trho^{R})^2\bigg] \\
&= \int \Tr\bigg[\big \{\Tilde{\mathcal{T}}(U^{A_1}\otimes U^{A_2})
\cdot \trho^{A_1A_2R})\big \}^2\bigg]-
\Tr\bigg[(\tomega^E\otimes \trho^{R})^2\bigg]
\end{align*}
By standard manipulations, using \cref{swaptrick} and the definition of 
adjoint of an operator, it follows that
\begin{align}
\int \Tr\bigg[\big\{\Tilde{\mathcal{T}}(U^{A_1}\otimes U^{A_2})\cdot 
\trho^{A_1A_2R})\big\}^2&\bigg]dU_1dU_2 \nonumber \\ 
& =\int\Tr\bigg[\big(\trho^{A_1A_2R}\big)^{\otimes 2}
\Big\{\Big(\big({U_1^{\dagger}}^{\otimes 2}\otimes 
{U_2^{\dagger}}^{\otimes 2}\big)\cdot 
{\Tilde{\mathcal{T}}^{\dagger}}^{\otimes 2}
\big(F^{EE'}\big)\Big)\bigotimes F^{RR'}\Big\}\bigg]dU_1dU_2 \nonumber \\
& =\Tr\bigg[\big(\trho^{A_1A_2R}\big)^{\otimes 2}
\Big\{\int\Big(\big({U_1^{\dagger}}^{\otimes 2}\otimes 
{U_2^{\dagger}}^{\otimes 2}\big)\cdot 
{\Tilde{\mathcal{T}}^{\dagger}}^{\otimes 2}\big(F^{EE'}\big)\Big)dU_1dU_2
\bigotimes F^{RR'}\Big\}\bigg] 
\label{eq:adjoint}
\end{align}
We will now use \cref{schurweylgeneral} by plugging in the matrix 
${\Tilde{\mathcal{T}}^{\dagger}}^{\otimes 2}\big(F^{EE'}\big)$ into $M$. 
The first step is to compute the entries of the vector on the R.H.S. 
in the matrix equation in \cref{schurweylgeneral}. We will demonstrate 
one such computation, the rest follow along similar lines.\\
\vspace{2mm}\\
\noindent 
\textbf{Computing $\Tr\bigg[ (F^{A_1A_1'}\otimes F^{A_2A'_2})~
{\Tilde{\mathcal{T}}^{\dagger}}^{\otimes 2}\big(F^{EE'}\big)\bigg]$ :}
\begin{align*}
\lefteqn{
\Tr\bigg[ (F^{A_1A_1'}\otimes F^{A_2A'_2})~
{\Tilde{\mathcal{T}}^{\dagger}}^{\otimes 2}\big(F^{EE'}\big)\bigg]} \\ 
&= 
\Tr\bigg[~\Tilde{\mathcal{T}}^{\otimes 2}\big( F^{A_1A_1'}\otimes 
F^{A_2A'_2}\big)~F^{EE'}\bigg] \\
&= 
\Tr\bigg[~\Tilde{\mathcal{T}}^{\otimes 2}\big( F^{A_1A_1'} \big(I^{A_1} 
\otimes I^{A_1'}\big)\otimes F^{A_2A'_2}\big(I^{A_2} \otimes I^{A_2'}
\big) \big)~F^{EE'}\bigg] \\
& \overset{a}{=} 
|A_1|^2 |A_2|^2\Tr\bigg[~\Tilde{\mathcal{T}}^{\otimes 2}
\big( F^{A_1A_1'} \big(\Tr_{\hat{A}_1}(\Phi^{A_1 \hat{A}_1}) 
\otimes \Tr_{\hat{A}_1'}(\Phi^{A_1' \hat{A}_1'})\big)\\
&\otimes F^{A_2A'_2}\big(\Tr_{\hat{A}_2}(\Phi^{A_2 \hat{A}_2}) 
\otimes \Tr_{\hat{A}_2'}(\Phi^{A_2' \hat{A}_2'})\big) \big)~F^{EE'}\bigg] 
\\
&= |A_1|^2 |A_2|^2\Tr\bigg[~\Tilde{\mathcal{T}}^{\otimes 2}
\big( \Tr_{\hat{A}_1,\hat{A}_1'} \big( (F^{A_1A_1'} \otimes 
I^{\hat{A}_1 \hat{A}_1'})(\Phi^{A_1 \hat{A}_1} \otimes 
\Phi^{A_1' \hat{A}_1'})\big)\\
&\otimes \Tr_{\hat{A}_2,\hat{A}_2'} \big( (F^{A_2A_2'} 
\otimes I^{\hat{A}_2 \hat{A}_2'})(\Phi^{A_2 \hat{A}_2} \otimes 
\Phi^{A_2' \hat{A}_2'})\big)~F^{EE'}\bigg] \\
&= 
|A_1|^2 |A_2|^2\Tr\bigg[~ \Tr_{\hat{A}_1,\hat{A}_1',\hat{A}_2,\hat{A}_2'} 
\big\{ \big( \Tilde{\mathcal{T}}^{\otimes 2}  \otimes 
\mathbb{I}^{\hat{A}_1,\hat{A}_1',\hat{A}_2,\hat{A}_2'} \big)
\big(  \big( (F^{A_1A_1'} \otimes I^{\hat{A}_1 \hat{A}_1'})(\Phi^{A_1 
\hat{A}_1 A_1' \hat{A}_1'})\big)\\
&\otimes \big( (F^{A_2A_2'} \otimes I^{\hat{A}_2 \hat{A}_2'})
(\Phi^{A_2 \hat{A}_2 A_2' \hat{A}_2'})\big) \big\}~F^{EE'}\bigg] \\
&= 
|A_1|^2 |A_2|^2\Tr\bigg[~ \big\{ \big( \Tilde{\mathcal{T}}^{\otimes 2}  
\otimes \mathbb{I}^{\hat{A}_1,\hat{A}_1',\hat{A}_2,\hat{A}_2'} \big)
\big(  \big( (F^{A_1A_1'} \otimes I^{\hat{A}_1 \hat{A}_1'})
(\Phi^{A_1 \hat{A}_1 A_1' \hat{A}_1'})\big)\\
&\otimes \big( (F^{A_2A_2'} \otimes I^{\hat{A}_2 \hat{A}_2'})
(\Phi^{A_2 \hat{A}_2 A_2' \hat{A}_2'})\big) \big\}~
\big( F^{EE'} \otimes I^{\hat{A}_1\hat{A}_1'} \otimes 
	I^{\hat{A}_2,\hat{A}_2'}\big)\bigg] \\
&\overset{b}{=} 
|A_1|^2 |A_2|^2\Tr\bigg[~ \big\{ \big( \Tilde{\mathcal{T}}^{\otimes 2}  
\otimes \mathbb{I}^{\hat{A}_1,\hat{A}_1',\hat{A}_2,\hat{A}_2'} \big)
\big(  \big( (I^{A_1A_1'} \otimes (F^T)^{~\hat{A}_1 \hat{A}_1'})
(\Phi^{A_1 \hat{A}_1 A_1' \hat{A}_1'})\big)\\
&\otimes \big( (I^{A_2A_2'} \otimes (F^T)^{~\hat{A}_2 \hat{A}_2'})
(\Phi^{A_2 \hat{A}_2 A_2' \hat{A}_2'})\big) \big\}~\big( F^{EE'} 
\otimes I^{\hat{A}_1\hat{A}_1'} \otimes 
I^{\hat{A}_2,\hat{A}_2'}\big)\bigg] \\
&= 
|A_1|^2 |A_2|^2\Tr\bigg[~ \big\{ \big( \Tilde{\mathcal{T}}^{\otimes 2}  
\otimes \mathbb{I}^{\hat{A}_1,\hat{A}_1',\hat{A}_2,\hat{A}_2'} \big)
\big( \Phi^{A_1 A_2 \hat{A}_1 \hat{A}_2} \otimes \Phi^{A_1' A_2' 
\hat{A}_1' \hat{A}_2'}\big) \times \\
&~\big( F^{EE'} \otimes (F^T)^{~\hat{A}_1 \hat{A}_1'} \otimes (F^T)^{
	~\hat{A}_1 \hat{A}_2'} \big)\bigg] \\
& = 
|A_1|^2 |A_2|^2 \Tr\bigg[~\Big(\big(\tomega^{\hat{A}_1\hat{A}'_1E} 
\otimes \tomega^{\hat{A}'_1\hat{A}'_2E'} \big) \big( F^{EE'} 
\otimes (F^T)^{~\hat{A}_1 \hat{A}_1'} 
\otimes (F^T)^{~\hat{A}_1 \hat{A}_2'} \big)\bigg] \\
&\overset{c}{=} 
|A_1|^2|A_2|^2 \norm{\tomega^{\hat{A}_1\hat{A}_2E}}_2^2
\end{align*}
where,
\begin{itemize}
\item{(a)} 
follows by defining systems $A_1 \cong \hat{A}_1'$, 
$A_2 \cong \hat{A}_2'$, $\Phi^{A_1 \hat{A}_1},\;
\Phi^{A_1' \hat{A}_1'},\;\Phi^{A_2 \hat{A}_2},\;
\Phi^{A_2' \hat{A}_2'}$ are the maximally entangled states and the fact 
that $I^{A_m}=|A_1|\Tr_{\hat{A}_m}(\Phi^{A_m \hat{A}_m})$ for $m=\{1,2\}$; 
\item{(b)} 
follows from the fact that for maximally entangled states, say  
$\Phi^{AA'}$ and any operator $M^{A}$ it holds that 
$(M^A \otimes I^{A'}) \Phi^{AA'} = (I^A \otimes (M^T)^{A'} )\Phi^{AA'}$ 
with the identification of the systems as 
$A=A_1 A_1', A'=\hat{A}_1 \hat{A}_1'$ and the operator $M$ as the 
swap operator $F$; and 
\item{(c)} 
follows from the \cref{swaptrick} and the observation that $F^T$ is 
equivalent to $F$.\\ 
\end{itemize}
Using similar arguments it can be shown that 
\begin{enumerate}
\item 
$\Tr\bigg[{\Tilde{\mathcal{T}}^{\dagger}}^{\otimes 2}\big(F^{EE'}\big) 
\bigg]=|A_1|^2|A_2|^2 \norm{\tomega^E}_2^2$
\item 
$\Tr\bigg[ F^{A_1A_1'}\otimes \mathbb{I}^{A_2A'_2}~{
\Tilde{\mathcal{T}}^{\dagger}}^{\otimes 2}\big(F^{EE'}\big)\bigg]
=|A_1|^2|A_2|^2\norm{\tomega^{\hat{A}_1E}}^2_2$
\item $\Tr\bigg[ \mathbb{I}^{A_1A_1'}\otimes F^{A_2A'_2}~{
\Tilde{\mathcal{T}}^{\dagger}}^{\otimes 2}\big(F^{EE'}\big)\bigg]
=|A_1|^2|A_2|^2\norm{\tomega^{\hat{A}_2E}}_2^2$
\end{enumerate}
Finally, to get meaningful bounds we need to bound the values of 
$\alpha_{00}, \alpha_{01}, \alpha_{10} \text{ and } \alpha_{11}$. To do 
this we first invert the matrix in \cref{schurweylgeneral} and observe 
the following:
\begin{align}
\label{eq:linearsim}
    \begin{bmatrix}
    \alpha_{00} \\ \alpha_{01} \\ \alpha_{10} \\ \alpha_{11}
    \end{bmatrix} 
&=\frac{\abs{A_1A_2}}{\big(\abs{A_1}^2-1\big)\big(\abs{A_2}^2-1\big)}
    \begin{bmatrix}
\abs{A_2} & -1 \\
-1 & \abs{A_2}
    \end{bmatrix} \bigotimes 
    \begin{bmatrix}
    \abs{A_1} & -1 \\
    -1 & \abs{A_1}
    \end{bmatrix}
    \begin{bmatrix}
    \norm{\tomega^E}_2^2 \\ 
    \norm{\tomega^{\hat{A}_1E}}^2_2 \\ 
    \norm{\tomega^{\hat{A}_2E}}_2^2 \\ 
    \norm{\tomega^{\hat{A}_1\hat{A}_2E}}_2^2
    \end{bmatrix} \\
    &=\frac{\abs{A_1A_2}}{\big(\abs{A_1}^2-1\big)\big(\abs{A_2}^2-1\big)}
    \begin{bmatrix}
    |A_1||A_2| & -\abs{A_2} & -\abs{A_1} & 1 \\
    -\abs{A_2} & |A_1||A_2| & 1 & -\abs{A_1} \\
    -\abs{A_1} & 1 & |A_1||A_2| & -\abs{A_2} \\
    1 & -\abs{A_1} & -\abs{A_2} & |A_1||A_2|
    \end{bmatrix}
    \begin{bmatrix}
    \norm{\tomega^E}_2^2 \\ 
    \norm{\tomega^{\hat{A}_1E}}^2_2 \\ 
    \norm{\tomega^{\hat{A}_2E}}_2^2 \\ 
    \norm{\tomega^{\hat{A}_1\hat{A}_2E}}_2^2
    \end{bmatrix}
\end{align}\\
\vspace{2mm}\\
\textbf{Bounding $\alpha_{00}$}

\noindent Consider the quantity 
\begin{align*}
\mathbf{I}:=
|A_1||A_2|\norm{\tomega^E}_2^2-
\abs{A_2}\norm{\tomega^{\hat{A}_1E}}^2_2-
\abs{A_1}\norm{\tomega^{\hat{A}_2E}}_2^2+
\norm{\tomega^{\hat{A}_1\hat{A}_2E}}_2^2
\end{align*}
By \cref{boundlem} we have that 
\begin{enumerate}
\item 
$\norm{\tomega^{\hat{A}_1\hat{A}_2E}}_2^2\leq 
\abs{A_1}\norm{\tomega^{\hat{A}_2E}}_2^2$
\item 
$\norm{\tomega^{\hat{A}_1E}}_2^2\geq 
\frac{1}{\abs{A_1}}\norm{\tomega^E}_2^2$
\end{enumerate}
The above inequalities imply that
\begin{align*}
\mathbf{I}\leq \frac{\abs{A_2}\cdot 
\big(\abs{A_1}^2-1\big)}{\abs{A_1}}\norm{\tomega^E}_2^2
\end{align*}
Thus on solving the system of Equations~\ref{eq:linearsim} and the 
bound on $\mathbf{I}$ further implies that
\begin{align*}
\alpha_{00} 
&= 
\frac{|A_1||A_2|}{\big(\abs{A_1}^2-1\big)\big(\abs{A_2}^2-1\big)} 
\mathbf{I} \\
&\leq 
\frac{\abs{A_2}^2}{\abs{A_2}^2-1} \norm{\tomega^E}_2^2
\end{align*}
\\
\vspace{2mm}\\
\textbf{Bounding $\alpha_{01}$}\\
\vspace{1mm}\\
Define 
\[
\mathbf{II}\coloneqq \begin{bmatrix} -\abs{A_2} & |A_1||A_2| & 1 & -\abs{A_1}\end{bmatrix} ~  \begin{bmatrix}
    \norm{\tomega^E}_2^2 \\ \norm{\tomega^{\hat{A}_1E}}^2_2 \\ \norm{\tomega^{\hat{A}_2E}}_2^2 \\ \norm{\tomega^{\hat{A}_1\hat{A}_2E}}_2^2
    \end{bmatrix}
\]
By \cref{boundlem} we have 
$\norm{\tomega^{\hat{A}_1E}}^2_2\leq \abs{A_1}\norm{\tomega^E}_2^2$ 
and 
$\norm{\tomega^{\hat{A}_1\hat{A}_2E}}_2^2\geq 
\abs{A_1}^{-1}\norm{\tomega^{\hat{A}_2E}}_2^2$, which implies that
\begin{align*}
\alpha_{01}
&=
\frac{|A_1||A_2|}{\big(\abs{A_1}^2-1\big)\big(\abs{A_2}^2-1\big)} 
\mathbf{II} \\
&\leq 
\frac{\abs{A_2}^2}{\abs{A_2}^2-1} \norm{\tomega^{\hat{A}_1E}}_2^2
\end{align*}
\textbf{Bounding $\alpha_{10}$}\\
\vspace{1mm}\\
Define 
\[
\mathbf{III}\coloneqq 
\begin{bmatrix} 
-\abs{A_1} & 1 & |A_1||A_2| & -\abs{A_2}
\end{bmatrix} ~  
\begin{bmatrix}
\norm{\tomega^E}_2^2 \\ 
\norm{\tomega^{\hat{A}_1E}}^2_2 \\ 
\norm{\tomega^{\hat{A}_2E}}_2^2 \\ 
\norm{\tomega^{\hat{A}_1\hat{A}_2E}}_2^2
\end{bmatrix}
\]
By \cref{boundlem} we have 
$\norm{\tomega^{\hat{A}_2E}}^2_2\leq 
\abs{A_1}\norm{\tomega^{\hat{A}_1\hat{A}_2E}}_2^2$ and 
$\norm{\tomega^{E}}_2^2\geq 
\abs{A_1}^{-1}\norm{\tomega^{\hat{A}_1E}}_2^2$, which implies that
\begin{align*}
\alpha_{10}
&=
\frac{|A_1||A_2|}{\big(\abs{A_1}^2-1\big)\big(\abs{A_2}^2-1\big)} 
\mathbf{III} \\
&\leq 
\frac{\abs{A_2}^2}{\abs{A_2}^2-1} \norm{\tomega^{\hat{A}_2E}}_2^2
\end{align*}
\textbf{Bounding $\alpha_{11}$}\\
\vspace{1mm}\\
Define 
\[
\mathbf{IV}=
\begin{bmatrix}
1 & -\abs{A_1} & -\abs{A_2} & |A_1||A_2|
\end{bmatrix} ~  
\begin{bmatrix}
\norm{\tomega^E}_2^2 \\ 
\norm{\tomega^{\hat{A}_1E}}^2_2 \\ 
\norm{\tomega^{\hat{A}_2E}}_2^2 \\ 
\norm{\tomega^{\hat{A}_1\hat{A}_2E}}_2^2
\end{bmatrix}
\]
By \cref{boundlem} we have 
$\norm{\tomega^{\hat{A}_2E}}^2_2\geq 
\abs{A_1}^{-1}\norm{\tomega^{\hat{A}_1\hat{A}_2E}}_2^2$ and 
$\norm{\tomega^{E}}_2^2\leq 
 \abs{A_1}\norm{\tomega^{\hat{A}_1E}}_2^2$, which implies that
\begin{align*}
\alpha_{11}
&=
\frac{|A_1||A_2|}{\big(\abs{A_1}^2-1\big)\big(\abs{A_2}^2-1\big)} 
\mathbf{IV} \\
&\leq 
\frac{\abs{A_2}^2}{\abs{A_2}^2-1} \norm{\tomega^{\hat{A}_1\hat{A}_2E}}_2^2
\end{align*}
Collating all these bounds and the fact that systems $\hat{A_i}$ can be 
relabelled as $A_i'$ for $i=\{1,~2\}$ we see that
\begin{align*}
\int \Tr\bigg[\big(&\Tilde{\mathcal{T}}
\{(U^{A_1}\otimes U^{A_2})\cdot \trho^{A_1A_2R}\}\big)^2\bigg]dU_1dU_2 \\
&\leq 
\frac{\abs{A_2}^2}{\abs{A_2}^2-1}\bigg[\norm{\trho^R}_2^2
\norm{\tomega^E}_2^2+\norm{\trho^{A_2R}}_2^2\norm{\tomega^{A'_2E}}_2^2+
\norm{\trho^{A_1R}}_2^2\norm{\tomega^{A'_1E}}_2^2+
\norm{\trho^{A_1A_2R}}_2^2\norm{\tomega^{A'_1A'_2E}}_2^2\bigg]
\end{align*}
Plugging this bound into the main expression of the theorem we find that
\begin{align*}
&\int\norm{\mathcal{T}(U^{A_1}\otimes U^{A_2}\cdot \rho^{A_1A_2R})-
\omega^E\otimes \rho^{R} }_1\\
&\leq 
\frac{1}{\abs{A_2}^2-1}\norm{\trho^R}_2^2\norm{\tomega^E}_2^2+
\frac{\abs{A_2}^2}{\abs{A_2}^2-1}
\bigg[\norm{\trho^{A_2R}}_2^2\norm{\tomega^{A'_2E}}_2^2+
\norm{\trho^{A_1R}}_2^2\norm{\tomega^{A'_1E}}_2^2+
\norm{\trho^{A_1A_2R}}_2^2\norm{\tomega^{A'_1A'_2E}}_2^2\bigg]
\end{align*}
By the choice of the weighting matrices $\sigma^E$ and $\zeta^R$ in the 
statement of the theorem, $\norm{\trho^R}_2^2\norm{\tomega^E}_2^2\leq 1$, 
for instance $\sigma^E$ is the matrix obtained by zeroing out the smallest
eigen values that sum up to $\delta$ and hence 
$\norm{\tomega^E}_2^2\leq 1$ and similarly choosing $\zeta^R$ by 
curtailing the marginal density operator $\rho^R$ and finally choosing 
$\abs{A_2}$ large such that the first term is less that $\delta$. This 
concludes the proof.
\end{proof}
We now state a corollary that will be useful in stating our coding 
theorems:
\begin{corollary}\label{codingcor}
Given the same conditions as in \cref{twodecoup} the following holds 
\begin{align*}
&\int\norm{\mathcal{T}\{(U^{A_1}\otimes U^{A_2})\cdot 
\rho^{A_1A_2R}\}-\omega^E\otimes \rho^{R}}_1dU^{A_1}dU^{A_2} \\
& \leq\Big(\delta+2 \cdot 2^{-\Tilde{H}_{2,\delta}(A_1|R)_{\rho}-
\Tilde{H}_{2,\delta}(A_1|E)_{\omega}}+
2^{-\Tilde{H}_{2,\delta}(A_2|R)_{\rho}-
\Tilde{H}_{2,\delta}(A_2|E)_{\omega}}+
2^{-\Tilde{H}_{2,\delta}(A_1A_2|R)_{\rho}-
   \Tilde{H}_{2,\delta}(A_1A_2|E)_{\omega}}\Big)
\end{align*}
\end{corollary}
\begin{proof}
The proof is easy, since the term 
$\frac{\abs{A}_2^2}{\abs{A}_2^2-1}\leq 2$ for all large $\abs{A_2}$. 
The rest follows trivially by the definition of 
$\Tilde{H}_{2,\delta}(\cdot|\cdot)$.
\end{proof}

\subsection{The Theorem for Multiple Senders}
We now generalise the tensor product decoupling theorem for $k>2$ senders.

\begin{theorem}[Generalised $k$ sender Tensor Product Decoupling Theorem]
\label{gen_decoup} 
Let $\rho^{A_0\ldots A_{k-1}R}$ be a density operator 
that can be thought of as an entangled state between $k$ senders with 
each sender denoted by $\{ A_i\}_{i=0}^ {k-1}$ and a reference $R$. Let  
$\mathcal{T}^{A_0\ldots A_{k-1}\rightarrow E}$ be a CP map, and define 
$\omega^{\hat{A}_0\ldots \hat{A}_{k-1}E}\coloneqq 
 (\mathbf{I}\otimes\mathcal{T})\cdot
 (\bigotimes\limits_{i}\Phi^{A_i\hat{A}_i})$, the Choi state for 
the superoperator $\mathcal{T}$. Then,
\begin{align*}
&\int\norm{\mathcal{T}\big
\{ \big(\bigotimes \limits_{i}U^{A_i}_{\textsc{rand}} \big) \cdot
\rho^{A_0\ldots A_{k-1}R}\big\}-\omega^{E}\otimes \rho^{R}}_1dU_0
\ldots dU_{k-1} \\
&\leq \Big(\frac{\abs{A_1\ldots A_{k-1}}^2}{(\abs{A_{k-1}}^2-1\ldots )
(\abs{A_1}^2-1)}-1\Big)\cdot\norm{\tomega^{E}}^2_2\cdot
\norm{\trho^R}_2^2  \\
&+\frac{\abs{A_0\ldots A_{k-1}}^2}{(\abs{A_{k-1}}^2-1\ldots )
(\abs{A_0}^2-1)}\sum\limits_{\mathbf{b}\neq 0^k} 
\norm{\tomega^{\mathbf{A^{\mathbf{b}}}E}}_2^2
\big(\norm{\trho^{\mathbf{A}^{\mathbf{b}}R}}_2^2\cdot 2^k+
\norm{\trho^R}_2^2\big)
\end{align*}
where we assume that $|A_0|$ is the smallest among the dimensions of the 
registers and the indices $\mathbf{b} \in \{0,1, \ldots, k-1 \}$ are 
represented as bit strings of length $k$.
\end{theorem}

\begin{proof}
For brevity of notation, let $A_{[k]}:=A_{k-1,k-2, \ldots, 1,0}$ denote 
the system representing the joint state of all the $k$ senders with 
dimension $|A_{[k]}|:=\Pi_{i=0}^{k-1} |A_i|$.
We will use the same definitions of $\tilde{\mathcal{T}},\tomega$ and 
$\trho$ as in \cref{twodecoup}, that is, to represent that 
$\tilde{\cdot}$ denotes conjugation of the underlying operator by 
appropriate weighting matrices arising due to Fact~\ref{cauchyschwartz}. 
We begin with the application of Fact~\ref{cauchyschwartz} as follows:
\begin{align*}
&\int\norm{\mathcal{T}\big\{ \big(\bigotimes 
\limits_{i}U^{A_i}_{\textsc{rand}}\big)\cdot\rho^{A_0\ldots A_{k-1}R}\big\}
-\omega^{E}\otimes \rho^{R}}_1dU_0\ldots dU_{k-1}\\
\leq 
& \sqrt{\int \Tr\Big[\Big(\tilde{\mathcal{T}}\big\{ \big(\bigotimes 
\limits_{i}U^{A_i}_{\textsc{rand}} \big)\cdot\trho^{
	A_0\ldots A_{k-1}R}\big\}-
\tomega^{E}\otimes \trho^{R}\Big)^2\Big]dU_0\ldots dU_{k-1}}
\end{align*}
Recall that
\begin{align*}
&\int \Tr\Big[\Big(\tilde{\mathcal{T}}\big
\{ \big(\bigotimes \limits_{i}U^{A_i}_{\textsc{rand}} \big)\cdot
\trho^{A_0\ldots A_{k-1}R}\big\}-\tomega^{E}\otimes \trho^{R}\Big)^2\Big]
dU_0\ldots dU_{k-1} \\
&=
\int\Tr\Big[\Big(\tilde{\mathcal{T}}\big\{\big(\bigotimes 
\limits_{i}U^{A_i}_{\textsc{rand}}\big)\cdot\trho^{A_0\ldots A_{k-1}R}\big
\}\Big)^2\Big]~d\big(\bigotimes\limits_{i}U_i\big)-\Tr[(\tomega^E)^2]
\cdot\Tr[(\trho^R)^2]
\end{align*}
Again, using standard arguments from Fact~\ref{swaptrick} and manipulation
similar to Equation~\ref{eq:adjoint} in \cref{twodecoup} we see that:
\begin{align}
\label{eq:mainexp}
&\int\Tr\Big[\Big(\tilde{\mathcal{T}}\big\{\big(\bigotimes 
\limits_{i}U^{A_i}_{\textsc{rand}}\big)\cdot\trho^{A_0\ldots A_{k-1}R}
\big\}\Big)^2\Big]d\big(\bigotimes\limits_{i}U_i\big) \nonumber \\
&= 
\Tr\Big[({\trho^{A_0\ldots A_{k-1}R}})^{\otimes 2}
\Big(\E\limits_{U_0^{\dagger A_0}, \ldots,U_{k-1}^{\dagger A_{k-1}} }
\big[\bigotimes\limits _{i}\big(U_i^{\dagger A_i}\otimes 
U_i^{\dagger A'_i}\big)\cdot M^{A_{[k]}~A'_{[k]}}\big]\otimes F^{RR'} 
\Big)\Big]
\end{align}
where $M^{A_{[k]}A'_{[k]}}\coloneqq \big(\tilde{\mathcal{T}}^{
\dagger}\big)^{ \otimes 2}\big(F^{EE'}\big)$ and the expectation is 
taken over independent choice of Haar random unitaries 
$\{ U_i \}_{i=0}^{k-1}$. \\
\vspace{2mm}\\
\noindent From \cref{schurweylgeneral} we have: 
\begin{align}
\label{eq:haarexp}
&\E\limits_{U_0^{\dagger A_0}, \ldots,U_{k-1}^{\dagger A_{k-1}} }
[\bigotimes\limits _{i}\big(U_i^{\dagger A_i}\otimes U_i^{\dagger A'_i}
\big)\cdot M^{A_{[k]}~A'_{[k]}}] \nonumber \\
=&
\E\limits_{U_0^{\dagger A_0}, \ldots,U_{k-1}^{\dagger A_{k-1}} }
[\bigotimes\limits _{i}\big(U_i^{A_i}\otimes U_i^{A'_i}\big)
\cdot M^{A_{[k]}~A'_{[k]}}] \\
&=
\sum\limits_{\mathbf{a}\coloneqq a_{k-1}\ldots a_0=0^k}^{1^k}
\alpha_{\mathbf{a}}\bigotimes\limits_{i}\Big(F^{A_iA'_i}\Big)^{a_i}
\end{align}
where in the last equality we represent the indices 
$\{\mathbf{a}\}_{\mathbf{a}=0}^{k-1}$ in binary as $a_{k-1},\ldots ,a_0$ 
for each $a_i \in {0,1}$ .
To evaluate the coefficients $\alpha_{\mathbf{a}}$ we again apply 
Lemma~\ref{schurweylgeneral} with the following equalities: 
\begin{align}
\label{eq:Tr(M)}
&\Tr (M^{A_{[k]} A_{[k]}^\prime})= 
\Tr  \left( \mathbb{E}_{\otimes_{i=0}^{k-1} U_i^{A_i}}
\left[ (\otimes_{i=0}^{k-1} (U_i^{ \dagger\; A_i} \otimes U_i^{\dagger \; 
A_i^\prime}))\circ (M^{A_{[k]} A_{[k]}^\prime}) \right] \right)    \\ 
\intertext{and} 
\label{eq:Tr(FM)}
& \Tr \left( (\otimes_{i=0}^{k-1} {(F^{A_i A_i^\prime})}^{a_i})(M^{A_{[k]}
A_{[k]}^\prime})  \right) \\
=&\Tr \left( (\otimes_{i=0}^{k-1} {(F^{A_i A_i^\prime})}^{a_i}) 
\mathbb{E}_{\otimes_{i=0}^{k-1} U_i^{A_i}}
\left[ (\otimes_{i=0}^{k-1} (U_i^{ \dagger\; A_i} \otimes 
	U_i^{\dagger \; A_i^\prime}))\circ (M^{A_{[k]} A_{[k]}^\prime}) 
\right] \right) \nonumber \\
.
\end{align}
This gives the matrix equation 
\begin{align}\label{mateq}
    K\cdot 
\begin{bmatrix}\alpha_{0^k} \\
    \vdots \\ \alpha_{\mathbf{a}} \\
    \vdots \\ \alpha_{1^k}
\end{bmatrix}=
|A_{[k]}|^2 
\begin{bmatrix} \vdots \\ 
\norm{\tomega^{\hat{\mathbf{A}}^{\mathbf{b}}E}}_2^2 \\ \vdots
\end{bmatrix}
\end{align}
where, for the bit string $\mathbf{b}\coloneqq b_{k-1}\ldots b_0$ 
we define:
\[\tomega^{\hat{\mathbf{A}}^{\mathbf{b}}E}\coloneqq 
\tomega^{\hat{A}_{k-1}^{b_{k-1}}\ldots \hat{A}_{0}^{b_{0}}E}
\]
The matrix $K$ is a $2^k\times 2^k$ matrix with rows indexed by bit 
vector $\mathbf{a} \in \brak{0,1}^k$ and columns indexed by the bit 
vector $\mathbf{b} \in\brak{0,1}^k$ and is obtained from \cref{eq:Tr(M)} 
and \cref{eq:Tr(FM)} with entries-1 
\begin{equation}
\label{mateq2}
(K)_{\mathbf{b},\mathbf{a}}=|A_{[k]}|\prod\limits_{i=0}^{k-1}
A_i^{\overline{(b_i \oplus a_i)}},
\end{equation}
where $\oplus$ denotes the bit-wise XOR. 
This is not hard to see as the $i$-th term in the product is the 
$(b_i,a_i)$ term of the $i$-th $2\times 2$ matrix in the tensor product 
in \cref{mateq}, which is exactly $\abs{A_i}^{\overline{b_i\oplus a_i}}$. 
\\ \vspace{2mm} \\
Also note that 
The RHS of \cref{mateq} comes from the fact that
\[
\Tr[\bigotimes\limits_{i}\big(F^{A_iA_i'}\big)^{b_i}M]=
\big(\prod\limits_{i\in \brak{0,1,\ldots k-1}}\abs{A_{i}}\Big)^2\norm{
\tomega^{\hat{\mathbf{A}}^{\mathbf{b}}E}}_2^2
\]
This leads to the following representation of $K$:
\begin{align}
\label{eq:matrixK}
K &= 
\abs{A_{k-1}}\abs{A_{k-2}}\ldots \abs{A_0} 
\begin{pmatrix} 
\abs{A_{k-1}} & 1 \\ 1 & \abs{A_{k-1}}
\end{pmatrix} \otimes 
\begin{pmatrix} \abs{A_{k-2}} & 1 \\ 1 & \abs{A_{k-2}}
\end{pmatrix}\otimes \ldots \otimes 
\begin{pmatrix} 
\abs{A_{0}} & 1 \\ 1 & \abs{A_{0}}
\end{pmatrix} \nonumber\\
 &=
\abs{A_{[k]}}\bigotimes\limits_{i\in \brak{k-1,k-2,\ldots, 0}}
\begin{pmatrix} 
\abs{A_i} & 1 \\ 1 & \abs{A_i} 
\end{pmatrix}
\end{align}
From \cref{eq:matrixK} we note that:
\begin{align}\label{mateq3}
K^{-1}=\frac{\abs{A_{[k]}}}{\big(\abs{A_{k-1}}^2-1\big)\ldots 
\big(\abs{A_0}^2-1\big)} \bigotimes\limits_{i\in 
\brak{k-1, k-2,\ldots,0}}
\begin{pmatrix} 
\abs{A_i} & -1 \\ -1 & \abs{A_i} 
\end{pmatrix}
\end{align}
Then coupled with \cref{mateq2} and \cref{mateq}, \cref{mateq3} implies 
that
\begin{align}
(K^{-1})_{\mathbf{a},\mathbf{b}}=
\frac{\abs{A_{[k]}}}{\big(\abs{A_{k-1}}^2-1\big)\ldots
\big(\abs{A_0}^2-1\big)}\prod\limits_{i\in \brak{k-1, k-2 \ldots 0}}
\abs{A_i}^{\overline{b_i\oplus a_i}}(-1)^{a_i\oplus b_i}
\end{align}
This directly implies that
\begin{align*}
\alpha_{\mathbf{a}}=
\frac{\abs{A_{[k]}}}{\big(\abs{A_{k-1}}^2-1\big)\ldots
\big(\abs{A_0}^2-1\big)}\sum\limits_{\mathbf{b}=0^k}^{1^k}
\Big(\prod\limits_{i\in \brak{k-1, k-2 \ldots 0}}
\abs{A_i}^{\overline{b_i\oplus a_i}}(-1)^{a_i\oplus b_i}\Big)
\norm{\tomega^{A_{[k]}^{\mathbf{b}}E}}_2^2
\end{align*}
We will differentiate between the following two cases. 
Define $\mathbf{c}\coloneqq \mathbf{a}\oplus \mathbf{b}$.
\subsection*{Case 1: $\mathbf{a}=0^k$}
and let without loss of generality $A_0$ be the register with the 
smallest dimension. Then, 
\begin{alignat}{2}
\label{eq:17}
\alpha_{\mathbf{a}} &= 
\frac{\abs{A_{k-1}\ldots A_0}^2}{\big(\abs{A_{k-1}}^2-1\big)\ldots 
\big(\abs{A_0}^2-1\big)}\Big( \norm{\tomega^E}_2^2-
\frac{\norm{\tomega^{A_0E}}_2^2}{\abs{A_0}}\Big) \nonumber \\
&+ \frac{\abs{A_{k-1}\ldots A_0}}{\big(\abs{A_{k-1}}^2-1\big)\ldots 
\big(\abs{A_0}^2-1\big)}\sum\limits_{\mathbf{c}'\coloneqq 
c_{k-1}\ldots c_1\neq 0^{k-1}}&\Big[\Big(\prod\limits_{i\neq 0}
\abs{A^{\bar{c}_i}_i}(-1)^{c_i}\Big)\abs{A_0}
\norm{\tomega^{\mathbf{A}^{\mathbf{c}'}E}}_2^2 \nonumber \\
&&-\Big(\prod\limits_{i\neq 0}\abs{A^{\bar{c}_i}_i}(-1)^{c_i}\Big)
\norm{\tomega^{\mathbf{A}^{\mathbf{c}'}A_0E}}_2^2\Big]
\end{alignat}
For $c'$ with odd parity and \cref{boundlem} the term inside the 
summation is:
\begin{align}
\label{eq:odd_c}
\Big(\prod_{i\neq 0}\abs{A^{\bar{c}_i}_i}\Big)
\Big[\norm{\tomega^{\mathbf{A}^{\mathbf{c}'}A_0E}}_2^2-
\abs{A_0}\norm{\tomega^{\mathbf{A}^{\mathbf{c}'}E}}_2^2\Big]\leq 0
\end{align}
For $c'$ with even parity and again \cref{boundlem} we have that:
\begin{align}
\label{eq:even_c}
&\Big(\prod\limits_{i\neq 0}\abs{A^{\bar{c}_i}_i}\Big)
\Big[\abs{A_0}\norm{\tomega^{\mathbf{A}^{\mathbf{c}'}E}}_2^2-
\norm{\tomega^{\mathbf{A}^{\mathbf{c}'}A_0E}}_2^2\Big] \nonumber \\
&\leq 
\Big(\prod\limits_{i\neq 0}\abs{A^{\bar{c}_i}_i}\Big) \cdot 
\Big(\frac{\abs{A_0}^2-1}{\abs{A_0}}\Big) \cdot 
\norm{\tomega^{\mathbf{A}^{\mathbf{c}'}E}}_2^2
\end{align}
Substituting \cref{eq:odd_c} and \cref{eq:even_c} in equation \cref{eq:17}
and using the bound \cref{boundlem} for 
$\norm{\tomega^{A_0E}}_2^2 \geq \norm{\tomega^{E}}_2^2/|A_0|$, we get:
\begin{align}
\label{eq:alpha_0_simplified}
\alpha_{0^k} \leq 
&\norm{\tomega^{E}}^2_2\cdot \frac{\abs{A_{[k]}}^2}
{(\abs{A_{k-1}}^2-1\ldots )(\abs{A_0}^2-1)}
\Big[1-\frac{1}{\abs{A_0}^2}\Big] \nonumber\\
&+\sum\limits_{\mathbf{c}'\neq
0^{k-1}}\norm{\tomega^{\mathbf{A}^{\mathbf{c}'}E}}^2_2\cdot
\frac{\abs{A_0}\cdot \prod\limits_{i\neq 0}\abs{A_i^{\bar{c}_i+1}}}
{(\abs{A_{k-1}}^2-1\ldots )(\abs{A_0}^2-1)}\Big[\abs{A_0}-
\frac{1}{\abs{A_0}}\Big] \nonumber\\
=&\norm{\tomega^{E}}^2_2\cdot \frac{\abs{A_1\ldots
A_{k-1}}^2}{(\abs{A_{k-1}}^2-1\ldots )(\abs{A_1}^2-1)}
\nonumber\\               
&+\sum\limits_{\mathbf{c}\neq
0^{k}, c_0=
0}\norm{\tomega^{\mathbf{A}^{\mathbf{c}}E}}^2_2\cdot
\frac{ \prod\limits_{i\neq 0}\abs{A_i^{\bar{c}_i+1}}}
	{(\abs{A_{k-1}}^2-1\ldots )(\abs{A_1}^2-1)}
\end{align} 

\subsection*{Case 2: $\mathbf{a}\neq 0^k$}
Firstly observe that, given a fixed $\mathbf{a}\in \brak{0,1}^k$, and 
$\mathbf{c}=\mathbf{a}\oplus \mathbf{b}$ for some 
$\mathbf{b}\in \brak{0,1}^k$, 
\begin{equation}
\label{eq:ubonXOR}    
\Big(\prod\limits_{i}\abs{A^{\bar{c}_i}_i}\Big)\cdot
\norm{\tomega^{\mathbf{A}^{\mathbf{a}\oplus\mathbf{c}}E}}_2^2 \leq
\Big(\prod\limits_{i}\abs{A_i}\Big)\cdot \norm{\tomega^{
	\mathbf{A}^{\mathbf{a}}E}}_2^2
\end{equation}

\noindent 
This is easy to verify on a case by case basis, by considering any fixed 
index $i\in [k]$ and iterating through all possible values of the tuple 
$(a_i,c_i)$. The above identity holds in each of these four possible 
cases, which is seen either directly or by invoking \cref{boundlem}, as 
the case demands.

\noindent 
Then we simply bound the value of $\alpha_{\mathbf{a}}$ as follows:
\begin{align}
\label{eq:alphaneq0}
\alpha_{\mathbf{a}} \leq 
&\frac{\abs{A_{[k]}}^2}{(\abs{A_{k-1}}^2-1\ldots )
(\abs{A_0}^2-1)}\norm{\tomega^{\mathbf{A}^{\mathbf{a}}E}}^2_2 \nonumber \\
&+\frac{\abs{A_{[k]}}}{(\abs{A_{k-1}}^2-1\ldots
)(\abs{A_0}^2-1)}\sum\limits_{\mathbf{b}\neq 0}\Big(\prod_{i}
\abs{A_i^{\bar{c}_i}}\Big)\cdot \norm{\tomega^{\mathbf{A}^{
\mathbf{a}\oplus\mathbf{c}}E}}_2^2 \nonumber \\
&\leq \frac{\abs{A_{[k]}}^2}{(\abs{A_{k-1}}^2-1\ldots )(\abs{A_0}^2-1)} 
	\norm{\tomega^{\mathbf{A}^{\mathbf{a}}E}}_2^2\cdot 2^k
\end{align}
where in the first inequality we upper bound every term by its absolute 
values and use \cref{eq:ubonXOR}. Finally we collate the estimates for 
$\alpha_{\mathbf{a}}$ for the two different cases from 
\cref{eq:alpha_0_simplified} and \cref{eq:alphaneq0} and substitute these 
values in \cref{eq:haarexp} to get:
\begin{align}
\label{eq:Haar_final}
&\mathbb{E}_{\otimes_{i=0}^{k-1} U_i^{A_i}}\left[ (\otimes_{i=0}^{k-1} 
(U_i^{ \dagger\; A_i} \otimes U_i^{\dagger \; A_i^\prime}))
\cdot M^{A_[k]A'_{k}} \right] \nonumber\\
&\leq \left( \frac{\abs{{(A_{k-1} \ldots A_1)}}^2}{(A_{k-1}^2-1) 
\ldots (A_1^2-1)} \lVert \tomega^E \rVert_2^2 \right) 
(I^{A_{[k]} A'_{[k]}}) \nonumber\\
&+ 
\left( \sum\limits_{\mathbf{c}\neq
0^{k}, c_0=
0}\norm{\tomega^{\mathbf{A}^{\mathbf{c}}E}}^2_2\cdot
\frac{ \prod\limits_{i\neq 0}\abs{A_i^{\bar{c}_i+1}}}
{(\abs{A_{k-1}}^2-1\ldots )(\abs{A_1}^2-1)} \right) (I^{A_{[k]} A'_{[k]}})
\nonumber\\
&+ \left(\sum\limits_{\mathbf{b}\neq 0^k}^{1^k} 
\frac{\abs{A_{[k]}}^2}{(\abs{A_{k-1}}^2-1\ldots )(\abs{A_0}^2-1)} 
\norm{\tomega^{\mathbf{A}^{\mathbf{a}}E}}_2^2\cdot 2^k \right) 
\bigotimes\limits_{i}\Big(F^{A_iA'_i}\Big)^{a_i}  
\end{align}
By substituting \cref{eq:Haar_final} in \cref{eq:mainexp} we get:
\begin{align*}
&\mathbb{E}_{U_0,\ldots ,U_{k-1}} \Tr\Big[\Big(\tilde{\mathcal{T}}
\big\{ \big(\bigotimes \limits_{i}U^{A_i}_{\textsc{rand}} \big)
\cdot\trho^{A_0\ldots A_{k-1}R}\big\}-\tomega^{E}\otimes \trho^{R}\Big)^2
\Big] \\ 
&\leq 
\Big(\frac{\abs{A_1\ldots A_{k-1}}^2}{(\abs{A_{k-1}}^2-1\ldots )
(\abs{A_1}^2-1)}-1\Big)\cdot\norm{\tomega^{E}}^2_2\cdot
\norm{\trho^R}_2^2  \\
&+\frac{\abs{A_0\ldots A_{k-1}}^2}{(\abs{A_{k-1}}^2-1\ldots )
(\abs{A_0}^2-1)}\sum\limits_{\mathbf{b}\neq 0^k} 
\norm{\tomega^{\mathbf{A^{\mathbf{b}}}E}}_2^2
\big(\norm{\trho^{\mathbf{A}^{\mathbf{b}}R}}_2^2\cdot 2^k+
\norm{\trho^R}_2^2\big)
\end{align*}
This concludes the proof.
\end{proof}
\begin{remark}
Just as in \cref{twodecoup} we can justify that 
$\lVert \tomega^E \rVert_2^2 \lVert \tilde{\rho}^R\rVert_2^2$ is much 
smaller than the second term in \cref{gen_decoup} above. For this, recall 
$\tomega^E := (\omega''_{\delta})^{-\frac{1}{4} \; E} \omega^E 
(\omega''_{\epsilon})^{-\frac{1}{4} \; E} $, with 
$(\omega''_{\epsilon})^{E} $ as the operator defined by zeroing out the 
smallest eigen values of $\omega^E$ that sum up to $\delta$. Thus,
$\lVert \tomega^E \rVert_2^2 \leq 1$.\\
Similarly, define $\zeta^R$ as the operator obtained by zeroing out those 
eigen values of $\rho^R$ that sum to $\delta$. Thus, 
$\lVert \trho^R \rVert_2^2 \leq 1$. Hence,
$
\Rightarrow \lVert \tomega^E \rVert_2^2\lVert \trho^R \rVert_2^2 \leq  1.
$
Thus the term with $\lVert \tomega^E \rVert_2^2\lVert \trho^R \rVert_2^2$ 
can be neglected in the multi-user tensor product decoupling theorem above.
The second term involving\\ 
$\norm{\tomega^{\mathbf{A}^{\mathbf{b}}E}}_2 \times 
\norm{\tilde{\rho}^{\mathbf{A}^{\mathbf{b}}E}}_2$ serves as the entropic 
quantity that gives the rate region for reliable communication for QMAC.  
\end{remark}

\section{The Multiple Access Channel}\label{QMAC}
In this section we will use the results of the previous section to derive 
coding theorems, in the one shot regime for the Quantum Multiple Access 
Channel. The task we will consider is the following : Given a quantum 
multiple access channel $\mathcal{N}^{A'B'\rightarrow C}$ with Alice and 
Bob as senders and Charlie as receiver, consider two pure states 
$\psi^{AC_1R_1}$ and $\varphi^{BC_2R_2}$ where $A$ and $B$ are registers 
that belong to Alice and Bob respectively, $C_1$ and $C_2$ belong to 
Charlie and $R_1$ and $R_2$ are purifying systems. The task is for Alice 
and Bob to send their shares of these states to Charlie by a single use 
of the channel. To do this we will show the existence of encoding 
isometries $\mathcal{U}_{\textsc{Alice}}$ and $\mathcal{V}_{\textsc{bob}}$
and a decodimg CPTP map $\mathcal{D}$ such that 
\[
\norm{\mathcal{D}\circ\mathcal{N}(\mathcal{U}_{\textsc{alice}}\otimes 
\mathcal{V}_{\textsc{bob}}\cdot \psi^{AC_1R_1}\otimes \varphi^{BC_2R_2})-
\psi^{AC_1R_1}\otimes \varphi^{BC_2R_2}}_1\leq \epsilon
\]
for some small $\epsilon$. The strategy we follow closely resembles the 
one for the point to point channel. 
\begin{theorem}\label{dupuismac}
Let $\psi^{AC_1R_1}$ and $\varphi^{BC_2R_2}$ be pure states and the 
registers $C_1$ and $C_2$ are held by Charlie. 
$\mathcal{N}^{A'B'\rightarrow C}$ is a CPTP map. 
$\omega^{A"B"CE}\coloneqq \mathcal{U}_{\mathcal{N}}(\Omega^{A'A"}\otimes 
\Delta^{B'B"})$ where $\Omega^{A'A"}$ and $\Delta^{B'B"}$ are pure states 
and $\abs{A"}=\abs{A'}$ and $\abs{B"}=\abs{B'}$. Then there exist encoding
isometries $\mathcal{U}_{\textsc{alice}}$ and $\mathcal{V}_{\textsc{bob}}$
and a decoding CPTP $\mathcal{D}^{CC_1C_2\rightarrow \hat{A}\hat{B}C_1C_2}$
such that
\begin{align*}
\norm{\mathcal{D}\circ\mathcal{N}((\mathcal{U}_{\textsc{alice}}\otimes 
\mathcal{V}_{\textsc{bob}})\cdot \psi^{AC_1R_1}\otimes \varphi^{BC_2R_2})-
\psi^{AC_1R_1}\otimes \varphi^{BC_2R_2}}_1\leq 2\sqrt{\delta_4}
\end{align*}
where, for some $\delta, \delta_1, \delta_2, \delta_3, \delta_4>0$, we 
have:
\begin{align*}
\delta +2^{-\Tilde{H}_{2,\delta}(A^"B^"|E)_{\omega}-
\Tilde{H}_{2,\delta}(A|R_1)_{\psi}-\Tilde{H}_{2,\delta}(B|R_2)_{\varphi}}+
2^{-\Tilde{H}_{2,\delta}(B^"|E)_{\omega}-
\Tilde{H}_{2,\delta}(B|R_2)_{\varphi}}+
&2^{-\Tilde{H}_{2,\delta}(A^"|E)_{\omega}-
\Tilde{H}_{2,\delta}(A|R_1)_{\psi}}
\coloneqq \delta^2_1 \\
2^{\frac{1}{2}H_{\max}(A)_{\psi}-
\frac{1}{2}\Tilde{H}_{2,\delta}(A")_{\omega}}&\coloneqq \delta_2 \\
2^{\frac{1}{2}H_{\max}(B)_{\varphi}-
\frac{1}{2}\Tilde{H}_{2,\delta}(B")_{\omega}}
&\coloneqq \delta_3\\
\delta_1+4\sqrt{\delta_2}+4\sqrt{\delta_3}+
12\sqrt{\delta_2\delta_3}&\coloneqq \delta_4
\end{align*}
\end{theorem}
\begin{proof}
Let the states $\ket{\Omega}^{\gol{A}A'}$ and $\ket{\Delta}^{\gol{B}B'}$ 
to be copies of the original $\ket{\Omega}^{A"A'}$ and 
$\ket{\Delta}^{A"A'}$ states, where $\abs{\gol{A}}=\abs{A"}$ and 
$\abs{\gol{B}}=\abs{B"}$.
Define 
\[
\mathcal{T}^{\gol{A}\gol{B}\rightarrow E}(\rho)\coloneqq
\abs{A"B"}~\bar{\mathcal{N}}(\mathrm{op}_{\gol{A}\rightarrow
A'}(\ket{\Omega})\otimes \mathrm{op}_{\gol{B}\rightarrow B'}
(\ket{\Delta})\cdot \rho)
\] 
Firstly, observe that :
\begin{align}
\omega^{A"B"E} &=  \Tr_C[\omega^{A"B"CE}]\\
&= \bar{\mathcal{N}}^{A'B'\rightarrow E}(\Omega^{A"A'}\otimes 
	\Delta^{B"B'}) \\
&= \mathcal{T}\otimes \mathbb{I}^{A"B"}(\Phi^{\gol{A}A"}\otimes 
	\Phi^{\gol{B}B"})
\end{align}

Let $W_1^{A\rightarrow \gol{A}}$ and $W_2^{B\rightarrow \gol{B}}$ be two 
isometries. Then the tensor product decoupling theorem then implies:
\begin{align*}
&\int\norm{\mathcal{T}((U_{\textsc{rand}}^{\gol{A}}W_1\otimes 
V_{\textsc{rand}}^{\gol{B}}W_2)\cdot (\psi^{AR_1}\otimes \varphi^{BR_2}))
- \omega^{E}\otimes \psi^{R_1}\otimes \varphi^{R_2}}_1 dU dV \\
&\leq \sqrt{\delta +2^{-\Tilde{H}_{2,\delta}(A^"B^"|E)_{\omega}-
\Tilde{H}_{2,\delta}(A|R_1)_{\psi}-\Tilde{H}_{2,\delta}(B|R_2)_{\varphi}}+
2^{-\Tilde{H}_{2,\delta}(B^"|E)_{\omega}-
\Tilde{H}_{2,\delta}(B|R_2)_{\varphi}}+
2^{-\Tilde{H}_{2,\delta}(A^"|E)_{\omega}-
\Tilde{H}_{2,\delta}(A|R_1)_{\psi}}} \\
    &= \delta_1
\end{align*}
Now we show the existence of two isometries $U_{\textsc{alice}}$ and
$V_{\textsc{bob}}$ which approximately emulate the action of the
operators $\sqrt{\abs{A"}}~\mathrm{op}_{\gol{A}\rightarrow
A'}(\Omega)U^{\gol{A}}_{\textsc{rand}}W_1$ and
$\sqrt{\abs{B"}}~\mathrm{op}_{\gol{B}\rightarrow B'}(\Delta)
V^{\gol{B}}_{\textsc{rand}}W_2$. To that end define the maps :
\begin{enumerate}
\item $\mathcal{E}^{\gol{A}\rightarrow G}(\rho)\coloneqq
\abs{A"}\Tr[\mathrm{op}_{\gol{A}\rightarrow A'}(\Omega)\cdot\rho]$
\item $\mathcal{F}^{\gol{B}\rightarrow G'}(\rho)\coloneqq
\abs{B"}\Tr[\mathrm{op}_{\gol{B}\rightarrow B'}(\Delta)\cdot\rho]$
\end{enumerate} 
where $G$ and $G'$ are one dimensional systems. Then, using the vanilla 
(non smooth) decoupling theorem twice we get
\begin{enumerate}
\item 
\[ 
\int \norm{\mathbb{I}^{C_1R_1}\otimes\mathcal{E}
(U_{\textsc{rand}}W_1\cdot \psi^{AC_1R_1})-\psi^{C_1R_1}}_1 dU
\leq 2^{\frac{1}{2}H_{\max}(A)_{\psi}-
\frac{1}{2}\Tilde{H}_{2,\delta}(A")_{\omega}}= \delta_2
\]
\item 
\[ 
\int \norm{\mathbb{I}^{C_2R_2}\otimes\mathcal{F}
(V_{\textsc{rand}}W_2\cdot \varphi^{BC_2R_2})-\varphi^{C_2R_2}}_1 dV
\leq 2^{\frac{1}{2}H_{\max}(B)_{\varphi}-
\frac{1}{2}\Tilde{H}_{2,\delta}(B")_{\omega}}= \delta_3
\]
\end{enumerate}
where we have used the facts that 
$\mathcal{E}(\Phi^{\gol{A}A"})=\omega^{A"}$ and 
$\mathcal{F}(\Phi^{\gol{B}B"})=\omega^{B"}$.  \\

Consider the random variables defined as follows:
\begin{enumerate}
\item 
$X\coloneqq \norm{\mathcal{T}(
(U_{\textsc{rand}}^{\gol{A}}W_1\otimes V_{\textsc{rand}}^{\gol{B}}W_2)
\cdot (\psi^{AR_1}\otimes \varphi^{BR_2})) - 
\omega^{E}\otimes \psi^{R_1}\otimes \varphi^{R_2}}_1$
\item 
$Y\coloneqq \norm{\mathbb{I}^{R_1}\otimes\mathcal{E}(
	U_{\textsc{rand}}W_1\cdot \psi^{AC_1R_1})-\psi^{C_1R_1}}_1$
\item 
$Z\coloneqq\norm{\mathbb{I}^{R_2}\otimes\mathcal{F}(V_{\textsc{rand}}
		W_2\cdot \varphi^{BC_2R_2})-\varphi^{C_2R_2}}_1
$
\end{enumerate}
and the following events:
\begin{enumerate}
    \item $E_1\coloneqq \{X\geq 4\delta_1\}$
    \item $E_2\coloneqq \{Y\geq 4\delta_2\}$
    \item $E_3\coloneqq \{Z\geq 4\delta_3\}$
\end{enumerate}
Now by Markov's inequality (for instance\\
$\Pr[E_2] \leq \big[ \frac{ \int  \norm{\mathbb{I}^{C_1R_1}\otimes
\mathcal{E}(U_{\textsc{rand}}W_1\cdot \psi^{AC_1R_1})-\psi^{C_1R_1}}_1 dU}
{4 \delta_2} 
\big]  \leq \frac{1}{4}$ ) and union bound for events $E_1,~E_2,E_3$ we get
\[
\Pr[\overline{E_1}\cap \overline{E_2}\cap \overline{E_3}]>0
\]
which implies that there exists fixed pair of unitaries 
$U_{\textsc{rand}}^{\gol{A}}$ and $V_{\textsc{rand}}^{\gol{B}}$ which 
satisfy the event $\overline{E_1}\cap \overline{E_2}\cap \overline{E_2}$. 
Fix such a pair of unitaries. Then, from Uhlmann's theorem we see that 
Fact~\ref{fact:Uhlmann}, there exist isometries 
$\mathcal{U}_{\textsc{alice}}^{A\rightarrow A'}$ and 
$\mathcal{V}_{\textsc{bob}}^{B\rightarrow B'}$ such that:
\begin{align}
& \norm{\abs{A"}\mathrm{op}_{\gol{A}\rightarrow A'}
(\Omega)U^{\gol{A}}_{\textsc{fixed}}W_1\cdot \psi^{AC_1R_1}-
\mathcal{U}_{\textsc{alice}}^{A\rightarrow A'}\cdot \psi^{AC_1R_1}}_1
\leq 4\sqrt{\delta_2} \\
& \norm{\abs{B"}\mathrm{op}_{\gol{B}\rightarrow B'}(\Delta)
V^{\gol{B}}_{\textsc{fixed}}W_2\cdot \varphi^{BC_2R_2}-
\mathcal{V}_{\textsc{bob}}^{B\rightarrow B'}\cdot \varphi^{BC_2R_2}}_1
\leq 4\sqrt{\delta_3}
\end{align}
Define $\Tr[\abs{B"}\mathrm{op}_{\gol{B}\rightarrow B'}
(\Delta)V^{\gol{B}}_{\textsc{fixed}}W_2\cdot \psi^{BC_2R_2}]\coloneqq c_0$.
Since trace is a quantum operation, from the equations above we see that 
\[
\abs{c_0-1}\leq 4\sqrt{\delta_3}
\]
This gives:
\begin{align}
&\norm{\abs{A"B"}(\mathrm{op}_{\gol{A}\rightarrow
A'}(\Omega)\mathcal{U}^{\gol{A}}_{\textsc{fixed}}W_1\otimes 
\mathrm{op}_{\gol{B}\rightarrow B'}(\Delta)
V^{\gol{B}}_{\textsc{fixed}}W_2)\cdot (\psi^{AC_1R_1}\otimes 
\varphi^{BC_2R_2})-\mathcal{U}_{\textsc{alice}}\otimes 
\mathcal{V}_{\textsc{bob}}\cdot \psi^{AC_1R_1}\otimes 
\varphi^{BC_2R_2}}_1 \\
\leq & \norm{\abs{B"}\mathrm{op}_{\gol{B}\rightarrow
B'}(\Delta)V^{\gol{B}}_{\textsc{fixed}}W_2\cdot \psi^{BC_2R_2}}_1\times
\norm{\abs{A"}\mathrm{op}_{\gol{A}\rightarrow A'}(\Omega)
U^{\gol{A}}_{\textsc{fixed}}W_1\cdot \psi^{AC_1R_1}-
\mathcal{U}_{\textsc{alice}}^{A\rightarrow A'}\cdot \psi^{AC_1R_1}}_1 \\
&+\norm{\mathcal{U}_{\textsc{alice}}^{A\rightarrow A'}\cdot
\psi^{AC_1R_1}}_1\times \norm{\abs{B"}
\mathrm{op}_{\gol{B}\rightarrow B'}(\Delta)
V^{\gol{B}}_{\textsc{fixed}}W_2\cdot \varphi^{BC_2R_2}-
\mathcal{V}_{\textsc{bob}}^{B\rightarrow B'}\cdot \varphi^{BC_2R_2}}_1 \\
\leq &(1+4\sqrt{\delta_3})\times (4\sqrt{\delta_2})+4\sqrt{\delta_3}
\end{align}
where we bound $c_0$ by $(1+4\sqrt{\delta_2})$. Finally, we use the 
triangle inequality and the monotonicity of $1$-norm under a quantum 
operation (which is partial trace over $C_1C_2$ followed by 
$\bar{\mathcal{N}}$) to obtain:
\begin{align}
\norm{\bar{\mathcal{N}}(\mathcal{U}_{\textsc{alice}}\otimes 
\mathcal{V}_{\textsc{bob}}\cdot \psi^{AR_1}\otimes \varphi^{BR_2})-
\omega^E\otimes \psi^{R_1}\otimes\varphi^{R_2}}_1
\leq &
\delta_1+4\sqrt{\delta_2}+4\sqrt{\delta_3}+12\sqrt{\delta_2\delta_3} \\
& = \delta_4
\end{align}
We conclude by invoking Uhlmann's theorem Fact~\ref{fact:Uhlmann} again 
for the last inequality to prove the exists a decoder 
$D^{CC_1C_2\rightarrow F\hat{A}\hat{B}C_1C_2}$ such that:
\[
\norm{D\mathcal{U}_{\bar{\mathcal{N}}}(\mathcal{U}_{\textsc{alice}}
\otimes \mathcal{V}_{\textsc{bob}}\cdot \psi^{AC_1R_1}\otimes 
\varphi^{BC_2R_2})-\lambda^{FE}\otimes \psi^{AC_1R_1}\otimes 
\varphi^{BC_2R_2}}_1\leq 2\sqrt{\delta_4}
\]
where $\lambda^{FE}$ is some purification of $\omega^E$.
\end{proof}
We are now ready to state our coding theorem. The task is to use the QMAC 
to send arbitrary states, tensored across the registers belonging to Alice
and Bob, with high fidelity to Charlie. It was shown in \cite{Tema} that 
this task is equivalent to sending one half of two maximally entangled 
states (one belonging to Alice and one to Bob) across the channel. For a
more general setting where there may be entanglement assistance, we 
reformulate the problem in the language of \cref{dupuismac}: \\
\vspace{1mm}\\
We set the states $\psi^{AC_1R_1}$ and $\varphi^{BC_2R_2}$ as 
$\Phi^{R_1M_1}\otimes \Phi^{\Tilde{A}C_1}$ and 
$\Phi^{R_2M_2}\otimes \Phi^{\Tilde{B}C_2}$ respectively. Here the 
registers $M_1\Tilde{A}$ and $M_2\Tilde{B}$ play the roles of $A$ and $B$.
For a given $\epsilon>0$ we say the the rate quadruple $(Q_A,E_A,Q_B,E_B)$
is $\epsilon$-achievable if there exist encoding isometries 
$\mathcal{U}_{\textsc{Alice}},~\mathcal{V}_{\textsc{Bob}}$ and decoding 
CPTP $\mathcal{D}$, with $\abs{M_1}=2^{Q_A}, \abs{\Tilde{A}}=2^{E_A}, 
\abs{M_2}=2^{Q_B}\text{ and } \abs{\Tilde{B}}=2^{E_B}$, such that
\[
\norm{\mathcal{D}\circ\mathcal{N}(\mathcal{U}_{\textsc{alice}}\otimes 
\mathcal{V}_{\textsc{bob}}\cdot \psi^{AC_1R_1}\otimes \varphi^{BC_2R_2})-
\psi^{AC_1R_1}\otimes \varphi^{BC_2R_2}}_1\leq \epsilon
\]
The rate pair $(Q_A,Q_B)$ is achievable for entangled assisted 
transmission if there exist $E_A,E_B\geq 0$ such that $(Q_A,Q_B,E_A,E_B)$ 
is $\epsilon$-achievable. The pair $(Q_A,Q_B)$ is achievable for 
unassisted transmission of $(Q_A,Q_B,0,0)$ is $\epsilon$-achievable. The 
one shot capacity region is the union of all achievable points $(Q_A,Q_B)$
for a fixed $\epsilon$, over all controlling states $\omega$ as defined in
\cref{dupuismac}.
\begin{theorem}
\label{thm:QMAC}
Given a quantum multiple access channel $\mathcal{N}^{A'B'\rightarrow C}$ 
and fixed $\delta>0$, define $\epsilon\coloneqq \delta_4$ where $\delta_4$
is as defined in \cref{dupuismac}. Let $\Omega^{A'A"}$ and $\Delta^{B'B"}$
be pure states and $\omega^{A"B"CE}\coloneqq 
\mathcal{U}_{\mathcal{N}}(\Omega\otimes \Delta)$. Then the rate quadruple 
$(Q_A,E_A,Q_B,E_B)$ is $\epsilon$-achievable for quantum transmission 
with rate limited entanglement assistance through $\mathcal{N}$ if
\begin{align*}
    Q_A-E_A+Q_B-E_B &< \Tilde{H}_{2,\delta}(A"B"|E)_{\omega}+
	2\log(1-\delta) \\
    Q_A-E_A &< \Tilde{H}_{2,\delta}(A"|E)_{\omega}+\log(1-\delta) \\
    Q_B-E_B &< \tilde{H}_{2,\delta}(B"|E)_{\omega}+\log(1-\delta)\\
    Q_A+E_A &<\Tilde{H}_{2,\delta}(A")_{\omega} \\
    Q_B+E_B &<\Tilde{H}_{2,\delta}(B")_{\omega}
\end{align*}
\end{theorem}
\begin{proof}
The proof is essentially an application of \cref{dupuismac}. First, set 
the states $\psi^{\Tilde{A}M_1C_1R_1}=\Phi^{R_1M_1}\otimes 
\Phi^{\Tilde{A}C_1}$ and 
$\varphi^{\Tilde{B}M_2C_2R_2}=\Phi^{R_2M_2}\otimes \Phi^{\Tilde{B}C_2}$ 
where the registers $\Tilde{A}M_1$ and $\Tilde{B}M_2$ are placeholders for
the registers $A$ and $B$ in \cref{dupuismac}. Then, invoking 
\cref{dupuismac} for the channel $\mathcal{N}$ with controlling state 
$\omega$, we see that there exist encoding isometries 
$\mathcal{U}_{\textsc{Alice}}, \mathcal{V}_{\textsc{bob}}$ and decoding 
CPTP $\mathcal{D}$ such that
\[
\norm{\mathcal{D}\circ\mathcal{N}(\mathcal{U}_{\textsc{alice}}\otimes 
\mathcal{V}_{\textsc{bob}}\cdot \psi^{\Tilde{A}M_1C_1R_1}\otimes 
\varphi^{\Tilde{B}M_2C_2R_2})-\psi^{\Tilde{A}M_1C_1R_1}\otimes 
\varphi^{\Tilde{B}M_2C_2R_2}}_1\leq \epsilon
\]
, where $\epsilon=\delta_4$ and 
\begin{align*}
\delta +2^{-\Tilde{H}_{2,\delta}(A^"B^"|E)_{\omega}-\Tilde{H}_{2,\delta}
(\Tilde{A}M_1|R_1)_{\psi}-\Tilde{H}_{2,\delta}(\Tilde{B}M_2|R_2)_{\varphi}}
+2^{-\Tilde{H}_{2,\delta}(B^"|E)_{\omega}-\Tilde{H}_{2,\delta}
(\Tilde{B}M_2|R_2)_{\varphi}}+&2^{-\Tilde{H}_{2,\delta}(A^"|E)_{\omega}-
\Tilde{H}_{2,\delta}(\Tilde{A}M_1|R_1)_{\psi}}
= \delta^2_1 \\
2^{\frac{1}{2}H_{\max}(\Tilde{A}M_1)_{\psi}-
\frac{1}{2}\Tilde{H}_{2,\delta}(A")_{\omega}}&= \delta_2 \\
2^{\frac{1}{2}H_{\max}(B)_{\varphi}-
\frac{1}{2}\Tilde{H}_{2,\delta}(B")_{\omega}}&= \delta_3\\
\delta_1+4\sqrt{\delta_2}+4\sqrt{\delta_3}+12\sqrt{\delta_2\delta_3}
&= \delta_4
\end{align*}
Observe that 
\begin{align*}
\Tilde{H}_{2,\delta}(\Tilde{A}M_1|R_1)_{\psi} 
\&\geq -Q_A+E_A+\log(1-\delta) \\
\Tilde{H}_{2,\delta}(\Tilde{B}M_2|R_2)_{\varphi} 
&\geq -Q_B+E_A+\log(1-\delta) \\
H_{\max}(\Tilde{A}M_1)_{\psi}& \leq Q_A+E_A \\
H_{\max}(\Tilde{B}M_2)_{\varphi}& \leq Q_B+E_B
\end{align*}
Plugging in these estimates in the expressions for $\delta_1, \delta_2, 
\delta_3$ such that $\delta_4=O(\sqrt{\delta})$ we conclude that the 
statement of the theorem is true.
\end{proof}

\section{Conclusion and Open Problems}
\label{Conclusion and open problem}
In this paper, we have proven a decoupling theorem which involves multiple
random unitaries, in tensor product with each other, chosen independently 
from the Haar measure as the decoupling unitary and the QMAC channel as the
superoperator that results in decoupling the reference system of the input
states and the channel environment, when expectation is taken over these 
unitaries in tensor product. The unitaries in tensor product can be thought
of as independent encoders and the decoupling is achieved as a decoding 
step. The analysis of the error rate in our decoupling theorem leads to the
characterization of an achievable rate region. We then proceed to evaluate
the asymptotic iid limit of our rate region.
However, we cannot recover the asymptotic iid rate region of Yard et al. 
in \cite{Yard_MAC}.
The reason being an immediate open problem, that is to find an optimising 
state that simultaneously smoothes the three different conditional 
R\'enyi $2$-entropies mentioned in Theorem~\ref{thm:QMAC}.

\bibliography{ref.bib}
\appendix 
\section{Appendix}
\newcommand{\h}[1]{#1'}
\newcommand{\hA}{\h{A}}
We present inner bounds for the $QMAC$ for the task of entanglement 
generation.

The following is an easy corollary of \cref{twodecoup} :
\begin{corollary}\label{Buscemi-style}
Given orthogonal projectors $\Pi_1^{A_1\rightarrow E_1}$ and 
$\Pi_2^{A_2\rightarrow E_2}$, define the map 
\[
\mathcal{T}^{A_1A_2\rightarrow E_1E_2}\coloneqq 
\frac{\abs{A_1}}{\abs{E_1}}\cdot \frac{\abs{A_2}}
{\abs{E_2}}\cdot \Pi_1^{A_1\rightarrow E_1}\otimes 
	\Pi_2^{A_2\rightarrow E_2}
\]
Then given the density operator $\rho^{A_1A_2R}$ the following holds:
\begin{align*}
&\int\norm{\mathcal{T}(U^{A_1}_{\textsc{rand}}\otimes 
U^{A_2}_{\textsc{rand}}\cdot \rho^{A_1A_2R})-\pi^{E_1}\otimes 
\pi^{E_2}\otimes \rho^{R}}_1dU^{A_1}dU^{A_2} \\
& \leq\Big(\delta+2 \cdot\abs{E_1}2^{
	-\Tilde{H}_{2,\delta}(A_1|R)_{\rho}}+
    \abs{E_2}2^{-\Tilde{H}_{2,\delta}(A_2|R)_{\rho}}+
    \abs{E_1E_2}2^{-\Tilde{H}_{2,\delta}(A_1A_2|R)_{\rho}}\Big)^{1/2}
\end{align*}
\end{corollary}
\begin{proof}
First, observe that 
\[
\Tr_{\hA_1\hA_2}\mathcal{T}(\Phi^{A_1\hA_1}\otimes \Phi^{A_2\hA_2})=
\pi^{E_1}\otimes \pi^{E_2}
\]
Then, define 
\[
\tomega^{\hA_1\hA_2E_1E_2}\coloneqq {\sigma^{E_1}}^{-1/4}\otimes 
{\sigma^{E_2}}^{-1/4}\cdot \mathcal{T}(\Phi^{A_1\hA_1}\otimes 
\Phi^{A_2\hA_2})
\]
where $\sigma^{E_1}\coloneqq \pi^{E_1}$ and $\sigma^{E_2}\coloneqq 
\pi^{E_2}$. Then note that:
\begin{enumerate}
\item $\tomega^{E_1}=\frac{\mathbb{I}^{E_1}}{\sqrt{\abs{E_1}}}$ and 
$\tomega^{E_2}=\frac{\mathbb{I}^{E_2}}{\sqrt{\abs{E_2}}}$
\item $\tomega^{\hA_1E_1E_2}=\tomega^{\hA_1E_1}\otimes \tomega^{E_2}$ 
and $\tomega^{\hA_2E_1E_2}=\tomega^{\hA_2E_2}\otimes \tomega^{E_1}$.
\end{enumerate}
It is easy to see that $\Tr[{\tomega^{\hA_1E_1}}^2]=\abs{E_1}$ since 
\begin{align}
\Tr[{\tomega^{\hA_1E_1}}^2] &= \Big(\sqrt{\abs{E_1}}\cdot 
\frac{\abs{A_1}}{\abs{E_1}}\Big)^2\cdot 
\Tr[\Pi_1\ket{\Phi}^{A_1\hA_1}\bra{\Phi}\Pi_1
	\ket{\Phi}^{A_1\hA_1}\bra{\Phi}\Pi_1] \\
&= \frac{\abs{A_1}^2}{\abs{E_1}}\Big(\bra{\Phi}\Pi_1\ket{\Phi}\Big)^2 \\
&= \frac{1}{\abs{E_1}} \Big(\Tr[\Pi_1]\Big)^2\\
&= \abs{E_1}
\end{align}
Similarly one can show that $\Tr[{\tomega^{\hA_2E_2}}^2]=\abs{E_2}$. We 
conclude by noting that 
$\norm{\tomega^{\hA_1E_1E_2}}_2^2=
\norm{\tomega^{\hA_1E_1}}_2^2\cdot \norm{\tomega^E_2}_2^2=\abs{E_1}$ and 
similarly $\norm{\tomega^{\hA_2E_1E_2}}^2_2=\abs{E_2}$. This completes 
the proof.
\end{proof}

To get a channel coding theorem we will use \cref{Buscemi-style}, but with
some overloading of notation. Consider the quantum multiple access channel
$\mathcal{N}^{A'B'\rightarrow C}$ with isometric extension 
$\mathcal{U}^{A'B'\rightarrow CE}$ where we use the register $E$ to mean 
the environment. Consider the controlling state 
$\ket{\omega}^{ABCE}\coloneqq 
\mathcal{U}^{A'B'\rightarrow CE}\ket{\Omega}^{AA'}\ket{\Delta}^{BB'}$, 
where $\ket{\Omega}^{AA'}$ and $\ket{\Delta}^{BB'}$ are arbitrary pure 
states. Then the following theorem holds:
\begin{theorem}
Given $\delta$ as in \cref{gen_decoup} and a positive 
$\epsilon$, the rates $(m_{\textsc{alice}},n_{\textsc{bob}})$ for 
entanglement generation over the channel $\mathcal{N}^{A'B'\rightarrow E}$
are achievable whenever
\begin{align*}
m_{\textsc{alice}} &< \Tilde{H}_{2,\delta}(A|E)_{\omega}+\log\epsilon \\
n_{\textsc{bob}} &< \Tilde{H}_{2,\delta}(B|E)_{\omega}+\log\epsilon \\
m_{\textsc{alice}}+n_{\textsc{bob}} 
&< \Tilde{H}_{2,\delta}(AB|E)_{\omega}+\log\epsilon
\end{align*}
with error $\sqrt{\delta+6\epsilon}$.
\end{theorem}

\begin{proof}
We simply relabel terms from \cref{Buscemi-style}:
\begin{enumerate}
    \item Registers: $A\gets A_1$, $B\gets B_1$ and $E\gets R$.
    \item Registers: $R_1\gets E_1$ and $R_2\gets E_2$.
    \item State: $\omega^{ABE}\gets \rho^{A_1A_2R}$.
    \item $\abs{R_1}=2^{m_{\textsc{alice}}}$ and 
	    $\abs{R_2}=2^{n_{\textsc{bob}}}$.
\end{enumerate}
Define 
\[
\mathcal{T}^{AB\rightarrow R_1R_2}\coloneqq 
\frac{\abs{A}}{\abs{R_1}}\cdot 
\frac{\abs{B}}{\abs{R_2}}\cdot \Pi_1^{A\rightarrow R_1}\otimes 
	\Pi_2^{B\rightarrow R_2}
\]
Then applying \cref{Buscemi-style} we see that:

\begin{align*}
&\int\norm{\mathcal{T}(U^{A}_{\textsc{rand}}\otimes 
U^{B}_{\textsc{rand}}\cdot \sigma^{ABE})-\pi^{R_1}\otimes 
\pi^{R_2}\otimes \rho^{E}}_1dU^{A}dU^{B} \\
& \leq\Big(\delta+2 \cdot \abs{R_1}2^{-\Tilde{H}_{2,\delta}(A|E)}+
\abs{R_2}2^{-\Tilde{H}_{2,\delta}(B|E)}+
\abs{R_1R_2}{2,\delta}^{-\Tilde{H}_2(AB|E)}\Big)^{1/2}
\end{align*}
Finally, the above equation implies that there exist fixed unitaries 
$\mathcal{U}^{A}_{\textsc{fixed}}$ and $\mathcal{U}^{B}_{\textsc{fixed}}$ 
such that the above inequality still holds. To conclude, by using the 
usual argument of applying Uhlmann's theorem to the purifying register 
$C$ and requiring that every term inside the curly braces be $<\epsilon$ 
we conclude the proof.
\end{proof}

\begin{remark}
We can emulate the $\mathcal{T}$ operation in the usual way by picking 
unitaries independently from using two $2$-designs instead of two Haar 
random unitaries. This reduces the required amount of shared randomness 
necessary to implement the $\mathcal{T}$ operation from infinite to the 
log of product of the cardinality of the designs. Classical communication 
is necessary however so that Alice and Bob can let Charlie know which 
code they are using.
\end{remark}

\end{document}